\documentclass[a4paper,11pt]{article}
\usepackage[a4paper, margin=1in]{geometry}
\usepackage[utf8]{inputenc}
\usepackage[T1]{fontenc}
\usepackage{amsmath, amssymb, amstext, amsthm, mathtools}
\usepackage{enumitem}
\usepackage{subcaption}
\usepackage{hyperref}
\usepackage{pgfplots}
\usepackage{tikz}

\definecolor[named]{urlblue}{cmyk}{1,0.58,0,0.21}

\hypersetup{
 breaklinks=true,
 colorlinks=true,
 citecolor=purple!70!blue!60!black,
 linkcolor=purple!70!blue!90!black,
 urlcolor=urlblue,
 pdflang={en},
 pdftitle={Isomorphism Testing for Graphs Excluding Small Minors},
 pdfauthor={Martin Grohe, Daniel Neuen and Daniel Wiebking}
}

\usetikzlibrary{backgrounds,patterns,arrows,automata,positioning,decorations,shapes,calc}
\usetikzlibrary{hobby}

\newcommand{\gurke}{
   \draw[draw=black,thick,fill=green,fill opacity=0.3,closed]
  (-1,0) .. (0,1) .. (1,0) .. (0.5,-2) .. (1,-4) .. (0,-5) .. (-1,-4) .. (-0.5,-2);
}

\newcommand{\randB}{
   \draw[draw=red,thick,dashed,closed]
  (-1.5,0) .. (0,1.5) .. (1.5,0) .. (1,-2) .. (1.5,-4.5) .. (0,-5.5) .. (-1.5,-4.5) .. (-1,-2);
}

\newcommand{\randA}{
   \draw[draw=red,thick,dashed,closed]
  (0,1.5) .. (2.5,-0.5) .. (4,-3) .. (0,-5.5) .. (-4,-3) .. (-2.5,-0.5);
}

\definecolor{darkpastelgreen}{rgb}{0.01, 0.75, 0.24}
\definecolor{royalpurple}{rgb}{0.47, 0.32, 0.66}
\definecolor{lime}{rgb}{0.75, 1.0, 0.0}
\definecolor{salmon}{rgb}{1.0, 0.55, 0.41}
\definecolor{orchid}{rgb}{0.85, 0.44, 0.84}
\definecolor{royalblue}{rgb}{0.0, 0.14, 0.4}
\definecolor{aquamarine}{rgb}{0.5, 1.0, 0.83}
\definecolor{wildstrawberry}{rgb}{1.0, 0.26, 0.64}

\definecolor{dandelion}{rgb}{0.94, 0.88, 0.19}
\definecolor{purple}{rgb}{0.63, 0.36, 0.94}
\definecolor{darkspringgreen}{rgb}{0.09, 0.45, 0.27}

\tikzstyle{color1}=[fill=red!80]
\tikzstyle{color2}=[fill=blue!80]
\tikzstyle{color3}=[fill=darkpastelgreen]
\tikzstyle{color4}=[fill=orange]
\tikzstyle{color5}=[fill=violet!80]
\tikzstyle{color6}=[fill=cyan!80]
\tikzstyle{color7}=[fill=yellow]
\tikzstyle{color8}=[fill=pink]
\tikzstyle{color9}=[fill=royalpurple]
\tikzstyle{color10}=[fill=lime]
\tikzstyle{color11}=[fill=salmon]
\tikzstyle{color12}=[fill=orchid]
\tikzstyle{color13}=[fill=royalblue]
\tikzstyle{color14}=[fill=aquamarine]
\tikzstyle{color15}=[fill=wildstrawberry]

\tikzstyle{lightcolor1}=[fill=red!30]
\tikzstyle{lightcolor2}=[fill=blue!30]
\tikzstyle{lightcolor3}=[fill=darkpastelgreen!30]
\tikzstyle{lightcolor4}=[fill=orange!30]
\tikzstyle{lightcolor5}=[fill=violet!30]
\tikzstyle{lightcolor6}=[fill=cyan!30]
\tikzstyle{lightcolor7}=[fill=yellow!30]
\tikzstyle{lightcolor8}=[fill=pink!40]
\tikzstyle{lightcolor9}=[fill=royalpurple!40]
\tikzstyle{lightcolor10}=[fill=lime!40]
\tikzstyle{lightcolor11}=[fill=salmon!40]
\tikzstyle{lightcolor12}=[fill=orchid!40]
\tikzstyle{lightcolor13}=[fill=royalblue!40]
\tikzstyle{lightcolor14}=[fill=aquamarine!40]
\tikzstyle{lightcolor15}=[fill=wildstrawberry!40]

\tikzstyle{emptyvertex}=[draw,circle,minimum size=7pt,inner sep=0pt]
\tikzstyle{midvertex}=[draw,fill=white,circle,minimum size=5pt,inner sep=0pt]
\tikzstyle{tinyvertex}=[draw,fill=white,circle,minimum size=3pt,inner sep=0pt]
\tikzstyle{thickedge}=[draw,gray!50,line width=4pt,-]

\newtheorem{theorem}{Theorem}[section]
\newtheorem{lemma}[theorem]{Lemma}

\theoremstyle{definition}
\newtheorem{definition}[theorem]{Definition}

\theoremstyle{remark}

\newtheorem{claim}{Claim}[section]

\newenvironment{claimproof}{\begin{proof}}{\end{proof}}

\newcommand{\WL}[2]{\chi_{\sf WL}^{#1}\![#2]}
\newcommand{\WLit}[3]{\chi_{#3,#2}^{#1}}

\newcommand{\CA}{{\mathcal A}}
\newcommand{\CB}{{\mathcal B}}

\newcommand{\CE}{{\mathcal E}}

\newcommand{\CH}{{\mathcal H}}

\newcommand{\CM}{{\mathcal M}}

\newcommand{\CO}{{\mathcal O}}
\newcommand{\CP}{{\mathcal P}}
\newcommand{\CQ}{{\mathcal Q}}

\newcommand{\CX}{{\mathcal X}}

\newcommand{\CZ}{{\mathcal Z}}

\newcommand{\NN}{{\mathbb N}}

\newcommand{\Fp}{{\mathfrak p}}

\renewcommand{\hat}{\widehat}

\DeclareMathOperator{\id}{id}

\DeclareMathOperator{\Iso}{Iso}
\DeclareMathOperator{\Aut}{Aut}
\DeclareMathOperator{\Sym}{Sym}

\DeclareMathOperator{\mgamma}{\widehat{\Gamma}}

\DeclareMathOperator{\cl}{cl}
\DeclareMathOperator{\atp}{atp}

\DeclareMathOperator{\fw}{fw}
\DeclareMathOperator{\bw}{bw}

\DeclareMathOperator{\argmin}{argmin}

\newcommand{\polylog}{\operatorname{polylog}}
\newcommand{\poly}{\operatorname{poly}}

\renewcommand{\phi}{\varphi}
\renewcommand{\epsilon}{\varepsilon}

\newcommand{\case}[1]{\par\medskip\noindent\textit{Case #1: }}

\newenvironment{cs}{
  \begin{description}
    \renewcommand{\case}[1]{\item[{\itshape\mdseries Case ##1:}]}
  }{
  \end{description}
}

\newcommand{\orcid}[1]{\href{https://orcid.org/#1}{\includegraphics[height=1.8ex]{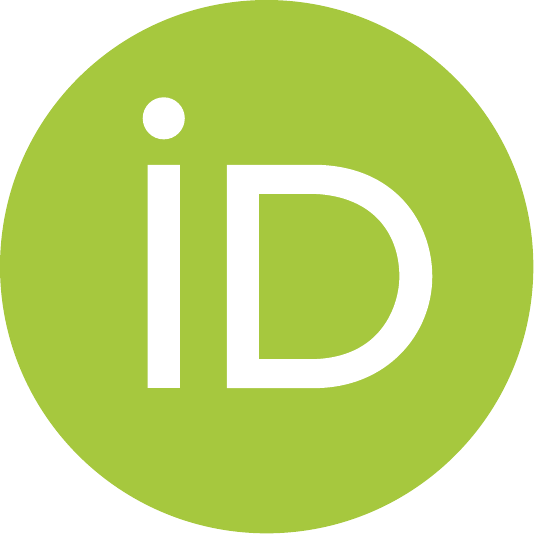}}}
\newcommand{\email}[1]{\href{mailto:#1}{\texttt{#1}}}

\title{Isomorphism Testing for Graphs Excluding Small Minors}
\author{
Martin Grohe \orcid{0000-0002-0292-9142}\\
RWTH Aachen University\\
\email{grohe@informatik.rwth-aachen.de}
\and
Daniel Neuen \orcid{0000-0002-4940-0318}\\
Simon Fraser University\\
\email{dneuen@sfu.ca}
\and
Daniel Wiebking\\
RWTH Aachen University\\
\email{wiebking@informatik.rwth-aachen.de}
}

\begin{document}

\maketitle

\begin{abstract}
 We prove that there is a graph isomorphism test running in time $n^{\operatorname{polylog}(h)}$ on $n$-vertex graphs excluding some $h$-vertex graph as a minor.
 Previously known bounds were $n^{\operatorname{poly}(h)}$ (Ponomarenko, 1988) and $n^{\operatorname{polylog}(n)}$ (Babai, STOC 2016).
 For the algorithm we combine recent advances in the group-theoretic graph isomorphism machinery with new graph-theoretic arguments.
\end{abstract}

\section{Introduction}

Determining the computational complexity of the Graph Isomorphism
Problem (GI) is one of best-known open problems in theoretical
computer science. The problem is obviously in NP, but neither known to
be NP-complete nor known to be solvable in polynomial time. In a
recent breakthrough result, Babai \cite{Babai16} presented a
quasipolynomial-time algorithm (i.e., an algorithm running in time
$n^{\polylog(n)}$) deciding isomorphism of two graphs, significantly
improving over the best previous algorithm running in time
$n^{\CO(\sqrt{n / \log n})}$ \cite{BabaiKL83}.  For his algorithm,
Babai greatly extends the group-theoretic isomorphism machinery dating
back to Luks \cite{Luks82} as well as our understanding of
combinatorial methods like the Weisfeiler-Leman algorithm (see, e.g.,
\cite{CaiFI92,WeisfeilerL68}).  Still, the question of whether the
Graph Isomorphism Problem can be solved in polynomial time remains
wide open.

Polynomial-time algorithms are known for restrictions of the Graph
Isomorphism Problem to several important graph classes
(e.g., \cite{FilottiM80,GroheM15,GroheS15,HopcroftT72,LokshtanovPPS17,Luks82,Miller80,Ponomarenko91}). In
particular, Luks~\cite{Luks82} gave an isomorphism algorithm running
in time $n^{\CO(d)}$ on input graphs of maximum degree $d$.  Building
on Luks's techniques and refinements due to Miller~\cite{Miller83a},
Ponomarenko~\cite{Ponomarenko91} designed an isomorphism test running
in time $n^{\poly(h)}$ for all graph classes that exclude a fixed
graph with $h$ vertices as a minor. Finally, it was shown that the
polynomial-time bound can be pushed to graph classes excluding a fixed
topological subgraph \cite{GroheM15}.

For the algorithms mentioned above the exponent of the polynomial always depends at least linearly on the parameter in question.
In light of Babai's quasi\-poly\-no\-mi\-al-time algorithm it seems natural to ask for which parameters these dependencies can be improved to polylogarithmic.

In \cite{GroheNS18} it was shown that Luks's original isomorphism test for bounded-degree graphs
can be combined with Babai's group-theoretic techniques.
By using a novel normalization technique, Schweitzer and the first two authors of this paper provided an isomorphism algorithm for graphs of maximum degree $d$ running in time $n^{\polylog(d)}$.
Recently, it was shown that the group-theoretic techniques used for bounded-degree graphs
can be extended to isomorphism testing of hypergraphs \cite{Neuen22b}.
This key subroutine finally led to an isomorphism test for graphs of Euler genus $g$ running in time $n^{\polylog(g)}$.
Another branch of research deals with the question how Babai's and Luks's group-theoretic techniques
can be combined with graph decomposition techniques \cite{Wiebking20} (see also \cite{GroheNSW20,SchweitzerW19}).
This series of papers finally led to an isomorphism test for graphs of tree-width at most $k$ running in time $n^{\polylog(k)}$.

In this work, we assemble the recent advances in the group-theoretic machinery developed in \cite{GroheNS18,Neuen22b,Wiebking20}
and combine it with new structural results for graphs with excluded
minors. Recall that a graph $H$ is a \emph{minor} of a graph $G$ if
$H$ is isomorphic to a graph that can be obtained from a subgraph of
$G$ by contracting edges. If $H$ is not a minor of $G$, we say that
$G$ \emph{excludes} $H$ as a minor. For example, all planar graphs
exclude the complete graph $K_5$ and the complete bipartite graph
$K_{3,3}$ as a minor, and in fact this characterizes the planar graphs
\cite{Wagner37}. Other natural classes of graphs excluding some fixed
graph as a minor are classes of bounded genus or bounded tree-width.

We present a new isomorphism test for graph classes that exclude a
fixed graph as a minor, improving the previously best algorithm for
this problem due to Ponomarenko \cite{Ponomarenko91} running in time $n^{\poly(h)}$.

\begin{theorem}
 There exists an algorithm deciding graph isomorphism in time $n^{\polylog(h)}$ on $n$-vertex graphs that exclude some $h$-vertex graph as a minor.
\end{theorem}

Note that a graph $G$ excludes some $h$-vertex graph as a minor if and
only if $G$ excludes the complete graph $K_h$ on $h$ vertices as minor.
Hence, for the remainder of this work, we restrict ourselves to the case where the input graphs exclude $K_h$ as minor.

The maximum $h$ such that $K_h$ is a minor of $G$ is known as the \emph{Hadwiger number} $\operatorname{hd}(G)$ of $G$.
Thus, an equivalent formulation of our result is that we design an isomorphism test for $n$-vertex
graphs running in time $n^{\polylog(\operatorname{hd}(G))}$.

Our proof heavily builds on the recently developed group-theoretic machinery
(the dependencies on the main previous results are shown in
Figure~\ref{fig:dep}). The main technical contributions of the present
paper are of a graph-theoretic nature. However, we are not using
Robertson-Seymour-style structure theory for graphs with excluded
minors~\cite{gmseries}, as one may expect given the previous results
for graphs of bounded genus and of bounded tree-width. Instead, our
results can be viewed as a structural theory for the automorphism
groups of such graphs; we find that graphs excluding $K_h$ as a minor
have an isomorphism-invariant decomposition into pieces whose
automorphism groups are similar to those of bounded-degree graphs
(Theorem~\ref{thm:structure-aut} is the precise statement). This
structural result may be of independent interest.
The only deeper graph-theoretic result we use is Kostochka's and
Thomason's theorem stating that graphs excluding $K_h$ as a minor have an
average degree of $\mathcal O(h \sqrt{\log h})$ \cite{Kostochka84,Thomason84}.

On a high level, our algorithm follows a decomposition strategy.
Given two graphs $G_1$ and $G_2$ excluding $K_h$ as a minor, the goal is to find isomorphism-invariant subsets $D_1\subseteq V(G_1)$ and $D_2\subseteq V(G_2)$ such that
one can control the interplay between the subsets and its complement and one can
significantly restrict the graph automorphisms on the two subsets.
Note that it is crucial to define the subsets $D_1$ and $D_2$ in an isomorphism-invariant fashion as to not compare two graphs that are decomposed in structurally different ways.
To capture the restrictions on the automorphism group, we build on the well-known class of $\mgamma_d$-groups, which are groups all whose composition factors are isomorphic to a subgroup of $S_d$ (the symmetric group on $d$ points).
However, to prove the restrictions on the automorphism group, we mostly use combinatorial and graph-theoretic arguments.

In particular, the algorithm heavily uses the $2$-dimensional Weisfeiler-Leman algorithm, a standard combinatorial algorithm which computes an isomorphism-invariant coloring of pairs of vertices.
In a lengthy case-by-case analysis depending on the color patterns computed by the $2$-dimensional Weisfeiler-Leman algorithm,
we are able to find initial isomorphism-invariant subsets $X_1\subseteq V(G_1) $ and $X_2\subseteq V(G_2)$ such that $(\Aut(G_i))_{v_i}[X_i]$ (the automorphism group of $G_i$ restricted to $X_i$ after fixing some vertex $v_i \in X_i$) forms a $\mgamma_t$-group where $t \in \CO(h^3 \log h)$.

In order to get control of the interplay between the subsets and their complement, we rely on a closure operator that builds on the notion of \emph{$t$-CR-bounded} graphs which were originally introduced by Ponomarenko in \cite{Ponomarenko89}\footnote{In \cite{Ponomarenko89} $t$-CR-bounded graphs are referred to as graphs with property $\Pi(0,t)$.}, and have been recently used to obtain an $n^{\polylog(g)}$ isomorphism test for graphs of Euler genus at most $g$ \cite{Neuen22b}.
Intuitively speaking, a graph $G$ is \emph{$t$-CR-bounded} if an initially uniform vertex-coloring $\chi$ can be turned into a discrete coloring (i.e., a coloring where every vertex has its own color) by repeatedly (a) applying the standard Color Refinement algorithm, and (b) splitting all color classes of size at most $t$.
We define the closure of a set $X_i$ (with respect to parameter $t$) to be the set $D_i$ of all vertices appearing in a singleton color class after (i) individualizing all vertices from the set $X_i$, and (ii) applying the above $t$-CR procedure.
This operator increases the subsets $X_1$ and $X_2$ in an isomorphism-invariant fashion
and leads to (possibly larger) sets $D_i \coloneqq \cl_t^{G_i}(X_i) \supseteq X_i$, $i \in \{1,2\}$.
A feature of this operator, which in a basic form was already observed in \cite{Weisfeiler76}, is that a given
$\mgamma_t$-group defined on the initial set $X_i$ can be extended to a $\mgamma_t$-group defined on the superset $D_i$ (see Theorem \ref{thm:compute-isomorphisms-t-closure}).
This provides us a $\mgamma_t$-group on the closure $D_i$ (after fixing a point) which allows the use of the group-theoretic techniques from \cite{GroheNS18,Neuen22b}.

The second main feature of the closure operator is that, in a graph $G$ that excludes an $h$-vertex graph as a minor,
the closure $D \coloneqq \cl_t^{G}(X)$ of any set $X \subseteq V(G)$ can only stop to grow at a separator of small size.
More precisely, we show that for every vertex set $Z$ of a connected component of $G - D$, it holds that $|N_{G}(Z)| < h$.
This key result shows that the interplay between $D$ and its
complement in $G$ is simple and allows for the application of the
group-theoretic decomposition framework from
\cite{GroheNSW20,SchweitzerW19,Wiebking20}.

We remark that our proof strategy is quite different from that used by
Ponomarenko \cite{Ponomarenko91} in his isomorphism test for graphs with excluded minors,
because we could not improve Miller's \cite{Miller83a} ``tower-of-$\mgamma_d$-groups''
technique to meet our quasipolynomial time demands.

\paragraph{Organization of the Paper.}
After introducing some basic preliminaries in the next section, we review the recent advances on the group-theoretic isomorphism machinery from \cite{Neuen22b,Wiebking20} in Section \ref{sec:group-machinery}.
Then, the main two technical theorems are presented in Section \ref{sec:exploit-minor}.
Finally, the complete algorithm is assembled in Section \ref{sec:isomorphism}.

\section{Preliminaries}
\label{sec:prelim}

\subsection{Graphs}

A \emph{graph} is a pair $G = (V(G),E(G))$ consisting of a \emph{vertex set} $V(G)$ and an \emph{edge set} $E(G) \subseteq
\binom{V(G)}{2}\coloneqq \big\{\{u,v\} \, \big\vert \, u,v \in V(G),u\neq v\big\}$.
All graphs considered in this paper are finite, undirected and simple (i.e., they contain no loops or multiple edges).
For $v,w \in V$, we also write $vw$ as a shorthand for $\{v,w\}$.
The \emph{neighborhood} of~$v$ is denoted by~$N_G(v)$.
The \emph{degree} of $v$, denoted by $\deg_G(v)$, is the number of edges incident with $v$, i.e., $\deg_G(v)=|N_G(v)|$.
For $X \subseteq V(G)$, we define $N_G(X) \coloneqq \left(\bigcup_{v \in X}N(v)\right) \setminus X$.
If the graph $G$ is clear from context, we usually omit the index and simply write $N(v)$, $\deg(v)$ and $N(X)$.

We write $K_n$ to denote the complete graph on $n$ vertices.
A graph is \emph{regular} if every vertex has the same degree.
A bipartite graph $G=(V_1,V_2,E)$ is called \emph{$(d_1,d_2)$-biregular} if all vertices $v_i \in V_i$ have degree $d_i$ for both $i \in \{1,2\}$.
In this case $d_1 \cdot |V_1| = d_2 \cdot |V_2| = |E|$.
By a double edge counting argument, for each subset $S \subseteq V_i$, $i\in\{1,2\}$, it holds that $|S| \cdot d_i \leq |N_G(S)| \cdot d_{3-i}$.
A bipartite graph is \emph{biregular}, if there are $d_1,d_2 \in \NN$ such that $G$ is $(d_1,d_2)$-biregular.
Each biregular graph satisfies the Hall condition, i.e., for all $S \subseteq V_1$ it holds $|S| \leq |N_G(S)|$ (assuming $|V_1| \leq |V_2|$).
Thus, by Hall's Marriage Theorem, each biregular graph contains a matching of size $\min(|V_1|,|V_2|)$.

A \emph{path of length $k$} from $v$ to $w$ is a sequence of distinct vertices $v = u_0,u_1,\dots,u_k = w$ such that $u_{i-1}u_i \in E(G)$ for all $i \in [k] \coloneqq \{1,\dots,k\}$.
For two sets $A,B\subseteq V(G)$, we denote by $G[A,B]$ the graph with vertex set $A\cup B$
and edge set $\{vw\in E(G)\mid v\in A,w\in B\}$.
For a set $A \subseteq V(G)$, we denote by $G[A]\coloneqq G[A,A]$ the \emph{induced subgraph} of $G$ on the vertex set $A$.
Also, we denote by $G - A$ the subgraph induced by the complement of $A$, that is, the graph $G - A \coloneqq G[V(G) \setminus A]$.
A graph $H$ is a \emph{subgraph} of $G$, denoted by $H \subseteq G$, if $V(H) \subseteq V(G)$ and $E(H) \subseteq E(G)$. 
A set $S \subseteq V(G)$ is a \emph{separator} of $G$ if $G - S$ has more connected components than $G$.
A \emph{$k$-separator} of $G$ is a separator of $G$ of size $k$.

An \emph{isomorphism} from $G$ to a graph $H$ is a bijection $\varphi\colon V(G) \rightarrow V(H)$ that respects the edge relation, that is, for all~$v,w \in V(G)$, it holds that~$vw \in E(G)$ if and only if $\varphi(v)\varphi(w) \in E(H)$.
Two graphs $G$ and $H$ are \emph{isomorphic}, written $G \cong H$, if there is an isomorphism from~$G$ to~$H$.
We write $\varphi\colon G\cong H$ to denote that $\varphi$ is an isomorphism from $G$ to $H$.
Also, $\Iso(G,H)$ denotes the set of all isomorphisms from $G$ to $H$.
The automorphism group of $G$ is $\Aut(G) \coloneqq \Iso(G,G)$.
Observe that, if $\Iso(G,H) \neq \emptyset$, it holds that $\Iso(G,H) = \Aut(G)\varphi \coloneqq \{\gamma\varphi \mid \gamma \in \Aut(G)\}$ for every isomorphism $\varphi \in \Iso(G,H)$.

A \emph{vertex-colored graph} is a tuple $(G,\chi)$ where $G$ is a graph and $\chi\colon V(G) \rightarrow C$ is a mapping into some set $C$ of colors, called \emph{vertex-coloring}.
Similarly, an \emph{arc-colored graph} is a tuple $(G,\chi)$, where $G$ is a graph and $\chi\colon\{(u,v) \mid \{u,v\} \in E(G)\} \rightarrow C$ is a mapping into some color set $C$, called \emph{arc-coloring}.
We also consider vertex- and arc-colored graphs $(G,\chi_V,\chi_E)$ where $\chi_V$ is a vertex-coloring
and $\chi_E$ is an arc-coloring.
Also, a \emph{pair-colored graph} is a tuple $(G,\chi)$, where $G$ is a graph and $\chi\colon (V(G))^{2} \rightarrow C$ is a mapping into some color set $C$.
Typically, $C$ is chosen to be an initial segment $[n]$ of the natural numbers.
Isomorphisms between vertex-, arc- and pair-colored graphs have to respect the colors of the vertices, arcs and pairs.

\subsection{Graph Minors and Topological Subgraphs}

Let $G$ be a graph.
A graph $H$ is a \emph{minor} of $G$ if $H$ can be obtained from $G$ by deleting vertices and edges of $G$ as well as contracting edges of $G$.
More formally, let $\CB = \{B_1,\dots,B_h\}$ be a partition of $V(G)$ such that $G[B_i]$ is connected for all $i \in [h]$.
We define $G/\CB$ to be the graph with vertex set $V(G/\CB) \coloneqq \CB$ and
\[E(G/\CB) \coloneqq \{BB' \mid \exists v \in B, v' \in B' \colon vv' \in E(G)\}.\]
A graph $H$ is a minor of $G$ if there is a partition $\CB = \{B_1,\dots,B_h\}$ of connected subsets $B_i \subseteq V(G)$ such that $H$ is isomorphic to a subgraph of $G/\CB$.
A graph $G$ \emph{excludes $H$ as a minor} if $H$ is not a minor of $G$.
The following theorem states the well-known fact that graphs excluding small minors have bounded average degree.
This was observed by Mader before Kostochka and Thomason independently proved an optimal bound on the average degree.

\begin{theorem}[\cite{Mader67,Kostochka84,Thomason84}]
 \label{thm:average-degree-excluded-minor}
 There is an absolute constant $ a \geq 1$ such that for every $h \geq 1$ and every graph $G$ that excludes $K_h$ as a minor, it holds that
 \[\frac{1}{|V(G)|}\sum_{v \in V(G)} \deg_G(v) \leq a h \sqrt{\log h}.\]
\end{theorem}

A graph $H$ is a \emph{topological subgraph} of $G$ if $H$ can be obtained from $G$ by deleting edges, deleting vertices and dissolving degree 2
vertices (which means deleting the vertex and making its two neighbors adjacent).
More formally, we say that $H$ is a topological subgraph of $G$ if a subdivision of $H$ is a subgraph of $G$
(a subdivision of a graph $H$ is obtained by replacing each edge of $H$ by a path of length at least 1).
Note that every topological subgraph of $G$ is also a minor of $G$.

\subsection{Weisfeiler-Leman Algorithm}

The Weisfeiler-Leman algorithm, originally introduced by Weisfeiler and Leman in its $2$-di\-men\-sion\-al form \cite{WeisfeilerL68}, forms one of the most fundamental subroutines in the context of isomorphism testing.
The algorithm presented in this work builds on the $1$-dimensional Weisfeiler-Leman algorithm, also known as the \emph{Color Refinement algorithm}, as well as the $2$-dimensional Weisfeiler-Leman algorithm.

Let~$\chi_1,\chi_2\colon V^k \rightarrow C$ be colorings of the~$k$-tuples of vertices of~$G$, where~$C$ is some finite set of colors. 
We say $\chi_1$ \emph{refines} $\chi_2$, denoted $\chi_1 \preceq \chi_2$, if $\chi_1(\bar v) = \chi_1(\bar w)$ implies $\chi_2(\bar v) = \chi_2(\bar w)$ for all $\bar v,\bar w \in V^{k}$.
The two colorings $\chi_1$ and $\chi_2$ are \emph{equivalent}, denoted $\chi_1 \equiv \chi_2$,  if $\chi_1 \preceq \chi_2$ and $\chi_2 \preceq \chi_1$.

The \emph{Color Refinement algorithm} (i.e., the $1$-dimensional Weisfeiler-Leman algorithm) is a procedure that, given a graph $G$, iteratively computes an isomorphism-invariant coloring of the vertices of $G$.
In this work, we actually require an extension of the Color Refinement algorithm that apart from vertex-colors also takes arc-colors into account.
We describe the mechanisms of the algorithm in the following.
For a vertex- and arc-colored graph $(G,\chi_V,\chi_E)$ we define $\WLit{1}{0}{G} \coloneqq \chi_V$ to be the initial coloring for the algorithm.
This coloring is iteratively refined by defining $\WLit{1}{i+1}{G}(v) \coloneqq (\WLit{1}{i}{G}(v), \CM_i(v))$ where
\[\CM_i(v) \coloneqq\big\{\!\!\big\{\big(\WLit{1}{i}{G}(w),\chi_E(v,w),\chi_E(w,v)\big) \;\big\vert\; w \in N_G(v) \big\}\!\!\big\}\]
(and $\{\!\!\{ \dots \}\!\!\}$ denotes a multiset).
By definition, $\WLit{1}{i+1}{G} \preceq \WLit{1}{i}{G}$ for all $i \geq 0$.
Thus, there is a minimal~$i$ such that $\WLit{1}{i+1}{G}$ is equivalent to $\WLit{1}{i}{G}$.
For this value of~$i$ we call the coloring~$\WLit{1}{i}{G}$ the \emph{stable} coloring of~$G$ and denote it by $\WL{1}{G}$.
The Color Refinement algorithm takes as input a vertex- and arc-colored graph $(G,\chi_V,\chi_E)$ and returns (a coloring that is equivalent to) $\WL{1}{G}$.
The procedure can be implemented in time $\CO((m+n)\log n)$ (see, e.g., \cite{BerkholzBG17}).

Next, we define the \emph{$2$-dimensional Weisfeiler-Leman algorithm}.
For a vertex-colored graph $(G,\chi_V)$ let $\WLit{2}{0}{G}\colon (V(G))^{2} \rightarrow C$ be the coloring where each pair is colored with the isomorphism type of its underlying ordered subgraph.
More formally, $\WLit{2}{0}{G}(v_1,v_2) = \WLit{2}{0}{G}(v_1',v_2')$ if and only if
$\chi_V(v_i)= \chi_V(v_i')$ for both $i \in \{1,2\}$, $v_1 = v_2 \Leftrightarrow v_1'= v_2'$ and $v_1v_2 \in E(G) \Leftrightarrow v_1'v_2' \in E(G)$.
We then recursively define the coloring~$\WLit{2}{i}{G}$ obtained after $i$ rounds of the algorithm.
Let $\WLit{2}{i+1}{G}(v_1,v_2) \coloneqq (\WLit{2}{i}{G}(v_1,v_2), \CM_i(v_1,v_2))$
where
\[\CM_i(v_1,v_2) \coloneqq \big\{\!\!\big\{\big(\WLit{2}{i}{G}(v_1,w),\WLit{2}{i}{G}(w,v_2)\big) \, \big\vert \, w \in V(G) \big\}\!\!\big\}.\]
Again, there is a minimal~$i$ such that $\WLit{2}{i+1}{G}$ is equivalent to $\WLit{2}{i}{G}$ and for this $i$ the coloring $\WL{2}{G} \coloneqq \WLit{2}{i}{G}$ is the \emph{stable} coloring of~$G$.

Note that the algorithm can easily be extended to arc-colored and pair-colored graphs by modifying the definition of the initial coloring $\WLit{2}{0}{G}$ accordingly.
However, in contrast to the Color Refinement algorithm, the $2$-dimensional Weisfeiler-Leman algorithm is only applied to vertex-colored graphs throughout this paper.

The $2$-dimensional Weisfeiler-Leman algorithm takes as input a (vertex-, arc- or pair-)colored graph $G$ and returns (a coloring that is equivalent to) $\WL{2}{G}$.
This can be implemented in time $\CO(n^3\log n)$ (see \cite{ImmermanL90}).

\subsection{Group Theory}

In this subsection, we introduce the group-theoretic notions required in this work.
For a general background on group theory we refer to \cite{Rotman99}, whereas background on permutation groups can be found in \cite{DixonM96}.

\paragraph{Permutation Groups.}

A \emph{permutation group} acting on a set $\Omega$ is a subgroup $\Gamma \leq \Sym(\Omega)$ of the symmetric group.
The size of the permutation domain $\Omega$ is called the \emph{degree} of $\Gamma$.
If $\Omega = [n]$, then we also write $S_n$ instead of $\Sym(\Omega)$.
For $\gamma \in \Gamma$ and $\alpha \in \Omega$ we denote by $\alpha^{\gamma}$ the image of $\alpha$ under the permutation $\gamma$.
The set $\alpha^{\Gamma} \coloneqq \{\alpha^{\gamma} \mid \gamma \in \Gamma\}$ is the \emph{orbit} of $\alpha$.

For $\alpha \in \Omega$ the group $\Gamma_\alpha \coloneqq \{\gamma \in \Gamma \mid \alpha^{\gamma} = \alpha\} \leq \Gamma$ is the \emph{stabilizer} of $\alpha$ in $\Gamma$.
The \emph{pointwise stabilizer} of a set $A \subseteq \Omega$ is the subgroup $\Gamma_{(A)} \coloneqq \{\gamma \in \Gamma \mid\forall \alpha \in A\colon \alpha^{\gamma}= \alpha \}$.
For $A \subseteq \Omega$ and $\gamma \in \Gamma$ let $A^{\gamma} \coloneqq \{\alpha^{\gamma} \mid \alpha \in A\}$.
The set $A$ is \emph{$\Gamma$-invariant} if $A^{\gamma} = A$ for all $\gamma \in \Gamma$.

For $A \subseteq \Omega$ and a bijection $\theta\colon \Omega \rightarrow \Omega'$ we denote by $\theta[A]$ the restriction of $\theta$ to the domain $A$.
For a $\Gamma$-invariant set $A \subseteq \Omega$, we denote by $\Gamma[A] \coloneqq \{\gamma[A] \mid \gamma \in \Gamma\}$ the induced action of $\Gamma$ on $A$, i.e., the group obtained from $\Gamma$ by restricting all permutations to $A$.
More generally, for every set $\Lambda$ of bijections with domain $\Omega$, we denote by $\Lambda[A] \coloneqq \{\theta[A] \mid \theta \in \Lambda\}$.

Let $\Gamma \leq \Sym(\Omega)$ and $\Gamma' \leq \Sym(\Omega')$.
A \emph{homomorphism} is a mapping $\varphi\colon \Gamma \rightarrow \Gamma'$ such that $\varphi(\gamma)\varphi(\delta) = \varphi(\gamma\delta)$ for all $\gamma,\delta \in \Gamma$.
A bijective homomorphism is also called \emph{isomorphism}.
For $\gamma \in \Gamma$ we denote by $\gamma^{\varphi}$ the $\varphi$-image of $\gamma$.
Similarly, for $\Delta \leq \Gamma$, we denote by $\Delta^{\varphi}$ the $\varphi$-image of $\Delta$ (note that $\Delta^{\varphi}$ is a subgroup of $\Gamma'$).

\paragraph{Algorithms for Permutation Groups.}

We review some basic facts about algorithms for permutation groups.
For detailed information we refer to \cite{Seress03}.

In order to perform computational tasks for permutation groups efficiently the groups are represented by generating sets of small size.
Indeed, most algorithms are based on so-called strong generating sets,
which can be chosen of size quadratic in the size of the permutation domain of the group and can be computed in polynomial time given an arbitrary generating set (see, e.g., \cite{Seress03}).

\begin{theorem}[cf.\ \cite{Seress03}] 
 \label{thm:permutation-group-library}
 Let $\Gamma \leq \Sym(\Omega)$ and let $S$ be a generating set for $\Gamma$.
 Then the following tasks can be performed in time polynomial in $n$ and $|S|$:
 \begin{enumerate}
  \item compute the order of $\Gamma$,
  \item given $\gamma \in \Sym(\Omega)$, test whether $\gamma \in \Gamma$,
  \item compute the orbits of $\Gamma$, and
  \item given $A \subseteq \Omega$, compute a generating set for $\Gamma_{(A)}$.
 \end{enumerate}
\end{theorem}

\paragraph{Groups with Restricted Composition Factors.}

In this work, we shall be interested in a particular subclass of permutation groups, namely groups with restricted composition factors.
Let $\Gamma$ be a group.
A \emph{subnormal series} is a sequence of subgroups $\Gamma = \Gamma_0 \trianglerighteq \Gamma_1 \trianglerighteq \dots \trianglerighteq \Gamma_k = \{\id\}$.
The length of the series is $k$ and the groups $\Gamma_{i-1} / \Gamma_{i}$ are the factor groups of the series, $i \in [k]$.
A \emph{composition series} is a strictly decreasing subnormal series of maximal length. 
For every finite group $\Gamma$ all composition series have the same family (considered as a multiset) of factor groups (cf.\ \cite{Rotman99}).
A \emph{composition factor} of a finite group $\Gamma$ is a factor group of a composition series of $\Gamma$.

\begin{definition}
 For $d \geq 2$ let $\mgamma_d$ denote the class of all groups $\Gamma$ for which every composition factor of $\Gamma$ is isomorphic to a subgroup of $S_d$.
\end{definition}

We want to stress the fact that there are two similar classes of
groups that have been used in the literature both typically denoted by~$\Gamma_d$.
One of these is the class introduced by Luks~\cite{Luks82}
that we denote by $\mgamma_d$, while the other one used in~\cite{BabaiCP82} in particular allows composition factors that are simple groups of Lie type of bounded dimension.

\begin{lemma}[Luks \cite{Luks82}]
 \label{la:gamma-d-closure}
 Let $\Gamma \in \mgamma_d$. Then
 \begin{enumerate}
  \item $\Delta \in \mgamma_d$ for every subgroup $\Delta \leq \Gamma$, and
  \item $\Gamma^{\varphi} \in \mgamma_d$ for every homomorphism $\varphi\colon \Gamma \rightarrow \Delta$.
 \end{enumerate}
\end{lemma}

\section{Group-Theoretic Techniques for Isomorphism Testing}
\label{sec:group-machinery}

Next, we present several group-theoretic tools in the context of
isomorphism testing which are exploited by our algorithm testing
isomorphism for graph classes that exclude a fixed minor. The
dependencies between the main results leading to this paper are shown
in Figure~\ref{fig:dep}.

\begin{figure}
 \centering
 \scalebox{0.83}{
 \begin{tikzpicture}[
  y = -2cm,
  every node/.style={draw, thick, text width = 5cm, text centered, rounded corners},
  every edge/.style={draw,thick,->}]
  
  \node (luks) at (3,0) {GI in time $n^{O(d)}$ \cite{Luks82}};
  \node (babai) at (3,1) {GI in quasipolynomial time \cite{Babai16}};
  \node (gns) at (0,2) {GI in time $n^{\polylog(d)}$ \cite{GroheNS18}};
  \node (lpps) at (9,1) {GI parameterized by tree-width in FPT \cite{LokshtanovPPS17}};
  \node (gnsw) at (9,2) {Group-theoretic GI test for bounded tree-width \cite{GroheNSW20}};
  \node (neuen) at (0,3) {Hypergraph Isomorphism for $\hat\Gamma_d$-groups \cite{Neuen22b}};
  \node (wiebking) at (6,3) {Decompositions with labeling cosets~\cite{SchweitzerW19,Wiebking20}};
  \node (this) at (3,4) {This paper};

  \draw (luks) edge (babai) (babai) edge (gns) edge (wiebking) (lpps)
   edge (gnsw) (gnsw) edge (wiebking) (gns) edge (neuen) (neuen) edge
   (this) (wiebking) edge (this);
 \end{tikzpicture}
 }
 \caption{Dependencies between the main results leading to this paper.}
 \label{fig:dep}
\end{figure}

\subsection{Hypergraph Isomorphism}

Two hypergraphs $\CH_1 = (V_1,\CE_1)$ and $\CH_2 = (V_2,\CE_2)$ are \emph{isomorphic} if there is a bijection $\varphi\colon V_1 \rightarrow V_2$ such that $E \in \CE_1$ if and only if $E^{\varphi} \in \CE_2$ for all $E \in 2^{V_1}$
(where $E^\phi\coloneqq\{\phi(v)\mid v\in E\}$ and $2^{V_1}$ denotes the power set of $V_1$).
We write $\varphi\colon \CH_1 \cong \CH_2$ to denote that $\varphi$ is an isomorphism from $\CH_1$ to $\CH_2$.
Consistent with previous notation, we denote by $\Iso(\CH_1,\CH_2)$ the set of isomorphisms from $\CH_1$ to $\CH_2$.
More generally, for $\Gamma \leq \Sym(V_1)$ and a bijection $\theta\colon V_1 \rightarrow V_2$, we define
\[\Iso_{\Gamma\theta}(\CH_1,\CH_2) \coloneqq \{\varphi \in \Gamma\theta \mid \varphi\colon \CH_1 \cong \CH_2\}.\]
The set $\Iso_{\Gamma\theta}(\CH_1,\CH_2)$ is either empty, or it is a coset of $\Aut_\Gamma(\CH_1) \coloneqq \Iso_\Gamma(\CH_1,\CH_1)$, i.e., $\Iso_{\Gamma\theta}(\CH_1,\CH_2) = \Aut_\Gamma(\CH_1)\varphi$ where $\varphi \in \Iso_{\Gamma\theta}(\CH_1,\CH_2)$ is an arbitrary isomorphism.
As a result, the set $\Iso_{\Gamma\theta}(\CH_1,\CH_2)$ can be represented efficiently by a generating set for $\Aut_\Gamma(\CH_1)$ and a single isomorphism $\varphi \in \Iso_{\Gamma\theta}(\CH_1,\CH_2)$.
In the remainder of this work, all sets of isomorphisms are represented in this way.

\begin{theorem}[{\cite[Theorem 1.1]{Neuen22b}}]
 \label{thm:hypergraph-isomorphism-gamma-d}
 Let $\CH_1 = (V_1,\CE_1)$ and $\CH_2 = (V_2,\CE_2)$ be two hypergraphs and let $\Gamma \leq \Sym(V_1)$ be a $\mgamma_d$-group and $\theta\colon V_1 \rightarrow V_2$ a bijection.
 Then $\Iso_{\Gamma\theta}(\CH_1,\CH_2)$ can be computed in time $(n+m)^{\CO((\log d)^{c})}$ for some absolute constant $c$ where $n \coloneqq |V_1|$ and $m \coloneqq |\CE_1|$.
\end{theorem}

\subsection{Coset-Labeled Hypergraphs}

Actually, for the applications in this paper, the Hypergraph Isomorphism Problem itself turns out to be insufficient.
Instead, we require a generalization of the problem that is, for example, motivated by graph decomposition approaches to graph isomorphism testing (see, e.g., \cite{GroheNSW20, Wiebking20}).
Let $G_1$ and $G_2$ be two graphs and suppose that an algorithm has already computed sets $D_1 \subseteq V(G_1)$ and $D_2 \subseteq V(G_2)$
in an isomorphism-invariant way, i.e., each isomorphism from $G_1$ to $G_2$ also maps $D_1$ to $D_2$.
Moreover, assume that $G_1 - D_1$ is not connected and
let $Z_1^{i},\dots,Z_\ell^{i}$ be the connected components of $G_i - D_i$ (without loss of generality $G_1 - D_1$ and $G_2 - D_2$ have the same number of connected components, otherwise the graphs are non-isomorphic).
Also, let $S_j^{i} \coloneqq  N_{G_i}(Z_{j}^{i})$ for all $j \in [\ell]$ and $i \in \{1,2\}$.
A natural strategy for an algorithm is to recursively compute representations for $\Iso(G_1[Z_{j_1}^{1} \cup S_{j_1}^{1}], G_2[Z_{j_2}^{2} \cup S_{j_2}^{2}])$ for all $j_1,j_2 \in [\ell]$.
Then, in the second step, the algorithm needs to compute all isomorphisms $\varphi\colon G_1[D_1] \cong G_2[D_2]$ such that there is a bijection $\sigma\colon [\ell] \rightarrow [\ell]$ satisfying
\begin{enumerate}
 \item $(S_{j}^{1})^{\varphi} = S_{\sigma(j)}^{2}$, and
 \item the restriction $\varphi[S_{j}^{1}]$ extends to an isomorphism from $G_1[Z_{j}^{1} \cup S_{j}^{1}]$ to  $G_2[Z_{\sigma(j)}^{2} \cup S_{\sigma(j)}^{2}]$ (in the natural way, i.e., the isomorphism restricted to $S_{j}^1$ equals $\varphi[S_{j}^{1}]$)
\end{enumerate}
for all $j \in [\ell]$.

Let us first discuss a simplified case where $S_{j_1}^{1} \neq S_{j_2}^{1}$ for all distinct $j_1,j_2 \in [\ell]$.
In this situation the first property naturally translates to an instance of the Hypergraph Isomorphism Problem (in particular, there is at most one bijection $\sigma$ for any given bijection $\varphi$).
However, for the second property, we also need to be able to put restrictions on how two hyperedges can be mapped to each other.
Towards this end, we consider hypergraphs with coset-labeled hyperedges where each hyperedge is additionally labeled by a coset.

A \emph{labeling} of a set $V$ is a bijection $\rho\colon V\to\{1,\ldots,|V|\}$.
A \emph{labeling coset} of a set $V$ is a set $\Lambda$ consisting of labelings such that
$\Lambda=\Delta\rho\coloneqq\{\delta\rho\mid \delta\in\Delta\}$ for some group $\Delta\leq\Sym(V)$ and some labeling $\rho\colon V\to\{1,\ldots,|V|\}$.
Observe that each labeling coset $\Delta\rho$ can also be written as
$\rho\Theta\coloneqq\{\rho\theta\mid\theta\in\Theta\}$ where
$\Theta\coloneqq\rho^{-1}\Delta\rho \leq S_{|V|}$.
  
\begin{definition}[Coset-Labeled Hypergraph]
 A \emph{coset-labeled hypergraph} is a tuple $\CH = (V,\CE,\mathfrak{p})$
 where $V$ is a finite set of vertices, $\mathcal{E} \subseteq 2^{V}$ is a set of hyperedges,
 and $\Fp$ is a function that associates with each $E \in \mathcal{E}$ a pair $\Fp(E) =(\rho\Theta,c)$
 consisting of a labeling coset of $E$ and a natural number $c \in \NN$.
 
 Two coset-labeled hypergraphs $\CH_1 = (V_1,\mathcal{E}_1,\Fp_1)$ and $\CH_2 = (V_2,\mathcal{E}_2,\Fp_2)$ are \emph{isomorphic} if there is a bijection $\phi\colon V_1 \rightarrow V_2$ such that
 \begin{enumerate}
  \item $E \in \CE_1$ if and only if $E^{\phi} \in \CE_2$ for all $E \in 2^{V_1}$, and
  \item for all $E \in\CE_1$ with $\Fp_1(E) = (\rho_1\Theta_1,c_1)$
    and  $\Fp_2(E^{\phi}) = (\rho_2\Theta_2,c_2)$ we have $c_1 = c_2$ and
   \begin{equation}
    \label{eq:coset-labeled-hypergraph}
    \phi[E]^{-1}\rho_1\Theta_1=\rho_2\Theta_2.
   \end{equation}
 \end{enumerate}
 In this case, $\phi$ is an \emph{isomorphism} from $\CH_1$ to $\CH_2$, denoted by $\phi\colon \CH_1 \cong \CH_2$.
 Observe that \eqref{eq:coset-labeled-hypergraph} is equivalent to $c_1 = c_2$, $\Theta_1=\Theta_2$ and $\phi[E] \in \rho_1\Theta_1\rho_2^{-1}$.
 For $\Gamma \leq \Sym(V_1)$ and a bijection $\theta\colon V_1 \rightarrow V_2$ let
 \begin{equation*}
  \Iso_{\Gamma\theta}(\CH_1,\CH_2) \coloneqq \{\phi \in \Gamma\theta \mid \phi\colon\CH_1 \cong \CH_2\}.
 \end{equation*}
\end{definition}

Note that, for two coset-labeled hypergraphs $\CH_1$ and $\CH_2$, the set of isomorphisms $\Iso(\CH_1,\CH_2)$ forms a coset of $\Aut(\CH_1)$ (or $\Iso(\CH_1,\CH_2) = \emptyset$) and therefore, it again admits a compact representation.
Indeed, this is a crucial feature of the above definition that again allows the application of group-theoretic techniques.

The next theorem is an immediate consequence of Theorem \ref{thm:hypergraph-isomorphism-gamma-d} and \cite[Theorem 6.6.7]{Neuen19}\footnote{In the notation of \cite[Theorem 6.6.7]{Neuen19}, there is a prototype $\Theta$ for every pair $(\Theta,c)$.
 We have $\Theta \in \mgamma_d$ since $|E| \leq d$ for all $E \in \CE_1 \cup \CE_2$.
 Finally, we use Theorem \ref{thm:hypergraph-isomorphism-gamma-d} to compute the induced coset of $\Iso_{\Gamma\theta}((V_1,\CE_1),(V_2,\CE_2))$ on the set of hyperedges $\CE_1$.}.

\begin{theorem}
 \label{thm:coset-labeled-hypergraphs-gamma-d}
 Let $\CH_1 = (V_1,\CE_1,\Fp_1)$ and $\CH_2 = (V_2,\CE_2,\Fp_2)$ be two coset-labeled hypergraphs
 such that for all $E\in \CE_1\cup\CE_2$ it holds $|E|\leq d$.
 Also let $\Gamma \leq \Sym(V_1)$ be a $\mgamma_d$-group and $\theta\colon V_1 \rightarrow V_2$ a bijection.
 
 Then $\Iso_{\Gamma\theta}(\CH_1,\CH_2)$ can be computed in time $(n+m)^{\CO((\log d)^{c})}$ for some absolute constant $c$ where $n \coloneqq |V_1|$ and $m \coloneqq |\mathcal{E}_1|$.
\end{theorem}

\subsection{Multiple-Labeling-Cosets}

The theorem above covers the problem discussed in the beginning of the previous subsection assuming that all separators of the first graph are distinct, i.e., $S_{j_1}^{1} \neq S_{j_2}^{1}$ for all distinct $j_1,j_2 \in [\ell]$.
In this subsection, we consider the case in which $S^1_{j_1}=S^1_{j_2}$ for all $j_1,j_2\in[\ell]$.
In order to handle the case of identical separators, we build on a framework considered in \cite{SchweitzerW19,Wiebking20}.
(The mixed case in which some, but not all, separators coincide can be handled by a mixture of both techniques.)

\begin{definition}[Multiple-Labeling-Coset]
 A \emph{multiple-labeling-coset} is a tuple $\CX=(V,L,\Fp)$ where $L=\{\rho_1\Theta_1,\ldots,\rho_t\Theta_t\}$
 is a set of labeling cosets $\rho_i\Theta_i$, $i \in [t]$, of the set $V$
 and $\Fp\colon L\to\NN$ is a function that assigns each labeling coset $\rho\Theta \in L$ a natural number $\Fp(\rho\Theta)=c$.
 
 Two multiple-labeling-cosets $\CX_1=(V_1,L_1,\Fp_1)$ and $\CX_2=(V_2,L_2,\Fp_2)$
 are \emph{isomorphic} if there is a bijection $\phi\colon V_1 \to V_2$ such that
 \begin{equation}
  \label{eq:multiple-labeling-cosets}
   \big(\;\rho\Theta \in L_1 \;\;\wedge\;\; \Fp_1(\rho\Theta) = c\;\big) \;\;\;\;\iff\;\;\;\; \big(\;\phi^{-1}\rho\Theta \in L_2 \;\;\wedge\;\; \Fp_2(\phi^{-1}\rho\Theta) = c\;\big)
 \end{equation}
 for all labeling cosets $\rho\Theta$ of $V$ and all $c \in \NN$.
 In this case, $\phi$ is an \emph{isomorphism} from $\CX_1$ to $\CX_2$, denoted by $\phi\colon\CX_1\cong\CX_2$.
 Observe that \eqref{eq:multiple-labeling-cosets} is equivalent to $|L_1| = |L_2|$ and for each $\rho_1\Theta_1\in L_1$ there is a $\rho_2\Theta_2\in L_2$
 such that $\Fp_1(\rho_1\Theta_1)=\Fp_2(\rho_2\Theta_2)$ and $\Theta_1=\Theta_2$ and $\phi\in\rho_1\Theta_1\rho_2^{-1}$.
 Let
 \begin{equation*}
  \Iso(\CX_1,\CX_2)\coloneqq\{\phi\colon V_1\to V_2\mid\phi\colon\CX_1\cong\CX_2\}
 \end{equation*}
\end{definition}

Again, the set of isomorphisms $\Iso(\CX_1,\CX_2)$ forms a coset of $\Aut(\CX_1) \coloneqq \Iso(\CX_1,\CX_1)$ (or $\Iso(\CX_1,\CX_2) = \emptyset$) and therefore, it again admits a compact representation.
The next theorem is obtained by a canonization approach building on the canonization framework from \cite{SchweitzerW19}.
Intuitively, a canonical form for a class of objects maps each object in that class to a representative of its isomorphism class.
For background on canonical forms and labelings we refer to \cite{BabaiL83}.

\begin{theorem}[{\cite[Theorem 22 and Corollary 35]{Wiebking20}}]
 Let $\CX=(V,L,\Fp)$ be a multiple-labeling-coset.
 Canonical labelings for $\CX$ can be computed in time $(n+m)^{\CO((\log n)^{c})}$ for some absolute constant $c$ where $n \coloneqq |V|$ and $m \coloneqq |L|$.
\end{theorem}

\begin{theorem}
 \label{thm:multiple-labeling-cosets-isomorphism}
 Let $\CX_1=(V_1,L_1,\Fp_1)$ and $\CX_2=(V_2,L_2,\Fp_2)$ be two multiple-labeling cosets.

 Then $\Iso(\CX_1,\CX_2)$ can be computed in time $(n+m)^{\CO((\log n)^{c})}$ for some absolute constant $c$ where $n \coloneqq |V_1|$ and $m \coloneqq |L_1|$.
\end{theorem}

\begin{proof}
 We compute canonical labelings $\Lambda_1,\Lambda_2$ for $\CX_1,\CX_2$, respectively.
 We compare the canonical forms $\CX_1^{\lambda_1}$ and $\CX_2^{\lambda_2}$ for labelings $\lambda_i\in\Lambda_i$, $i \in \{1,2\}$
 (this can be done in polynomial time as shown in \cite{Wiebking20}).
 We can assume that the canonical forms are equal, otherwise we reject isomorphism.
 Then, we return $\Iso(\CX_1,\CX_2) = \Lambda_1\lambda_2^{-1}$.
\end{proof}

\subsection{Allowing Color Refinement to Split Small Color Classes}

In order to be able to apply the decomposition framework outlined above, an algorithm first needs to compute an isomorphism-invariant subset $D \subseteq V(G)$ such that $N_G(Z)$ is sufficiently small for every connected component $Z$ of the graph $G - D$.
Moreover, the application of Theorem \ref{thm:coset-labeled-hypergraphs-gamma-d} additionally requires a $\mgamma_d$-group that restricts the set of possible automorphisms for the set $D$.
Both problems are tackled building on the notion of \emph{$t$-CR-bounded graphs}.
This class of graphs was originally introduced by Ponomarenko in the late 1980's \cite{Ponomarenko89} and has been exploited more recently for isomorphism testing of graphs of bounded genus \cite{Neuen22b} which form an important subfamily of graph classes excluding a fixed graph as a minor.

Intuitively speaking, a vertex-colored graph $(G,\chi)$ is $t$-CR-bounded, $t \in \NN$, if it is possible to obtain
a discrete vertex-coloring (a vertex-coloring is \emph{discrete} if each vertex has a distinct color) for the graph by iteratively applying the following two operations:
\begin{itemize}
 \item applying the Color Refinement algorithm, and
 \item picking a color class $[v]_\chi \coloneqq \{w \in V(G) \mid \chi(v) = \chi(w)\}$ for some vertex $v \in V(G)$ where $|[v]_\chi|\leq t$ and individualizing each vertex in that class (every vertex in that color class is assigned a distinct color).
\end{itemize}

In this work, we exploit the ideas behind $t$-CR-bounded graphs to define a closure operator.
Given an initial set $X \subseteq V(G)$, all vertices from $X$ are first individualized before applying the operators of the $t$-CR-bounded definition.
The closure of the set $X$ (with respect to the parameter $t$) then contains all singleton vertices of the resulting coloring.
The next definition formalizes all these concepts.
Since we usually deal with vertex- and arc-colored graphs, the definition is formulated in this general setting.

\begin{definition}
 \label{def:t-cr-bounded}
 Let $(G,\chi_V,\chi_E)$ be a vertex- and arc-colored graph and $X \subseteq V(G)$.
 Let $(\chi_i)_{i \geq 0}$ be the sequence of vertex-colorings where
 \[\chi_0(v) \coloneqq \begin{cases}
                (v,1)         & \text{if } v \in X\\
                (\chi_V(v),0) & \text{otherwise}
               \end{cases},\]
 $\chi_{2i+1} \coloneqq \WL{1}{G,\chi_{2i},\chi_E}$ and
 \[\chi_{2i+2}(v) \coloneqq \begin{cases}
                     (v,1)              & \text{if } |[v]_{\chi_{2i+1}}| \leq t\\
                     (\chi_{2i+1}(v),0) & \text{otherwise}
                    \end{cases}\]
 for all $i \geq 0$.
 Since $\chi_{i+1} \preceq \chi_{i}$ for all $i \geq 0$ there is some minimal $i^{*}$ such that $\chi_{i^{*}} \equiv \chi_{i^{*}+1}$.
 We define
 \[\cl_t^{(G,\chi_V,\chi_E)}(X) \coloneqq \left\{v \in V(G) \mid \left|[v]_{\chi_{i^{*}}}\right| = 1\right\}.\]
 For $v_1,\dots,v_k \in V(G)$ we also denote
 \[\cl_t^{(G,\chi_V,\chi_E)}(v_1,\dots,v_k) \coloneqq \cl_t^{(G,\chi_V,\chi_E)}(\{v_1,\dots,v_k\}).\]
 
 Moreover, the pair $(G,X)$ is \emph{$t$-CR-bounded} if $\cl_t^{(G,\chi_V,\chi_E)}(X) = V(G)$.
 Finally, the graph $G$ is \emph{$t$-CR-bounded} if $(G,\emptyset)$ is $t$-CR-bounded.
\end{definition}

For ease of notation, we usually omit the vertex- and arc-colorings and simply write $\cl_t^{G}$ instead of $\cl_t^{(G,\chi_V,\chi_E)}$.

For applications in graph classes with an excluded minor it turns out to be useful to combine the concept of $\cl_t^{G}$ with the $2$-dimensional Weisfeiler-Leman algorithm.
More precisely, in order to increase the scope of the set $\cl_t^{G}$, information computed by the $2$-dimensional Weisfeiler-Leman algorithm are taken into account.
Since the $2$-dimensional Weisfeiler-Leman algorithm computes a pair-coloring, we extend the definition of $\cl_t^{G}$ to pair-colored graphs.
For a pair-colored graph $(G,\chi)$ we define $\cl_t^{(G,\chi)} \coloneqq \cl_t^{(K_n,\widetilde\chi)}$ where
$K_n$ is the complete graph on the same vertex set $V(G)$ and $\widetilde\chi(v,w) = (\atp(v,w),\chi(v,w))$ where $\atp(v,w) = 0$ if $v=w$, $\atp(v,w) = 1$ if $vw \in E(G)$, and $\atp(v,w) = 2$ otherwise.
This allows us to take all pair-colors into account for the Color Refinement algorithm, but also still respect the edges of the input graph $G$.

It can be shown that for each $t$-CR-bounded graph $G$ it holds that $\Aut(G) \in \mgamma_t$.
Moreover, there is an algorithm that, given a graph $G$, computes a $\mgamma_t$-group $\Gamma \leq \Sym(V(G))$ such that $\Aut(G) \leq \Gamma$ in time $n^{\polylog(t)}$ where $n$ is the number of vertices of $G$.
It is important for our techniques that this statement generalizes to $t$-CR-bounded pairs $(G,X)$
for which we already have a good knowledge of the structure of $X$ in form of a $\mgamma_t$-group $\Gamma \leq \Sym(X)$
as stated in the following theorem.

\begin{theorem}[{\cite[Lemma 5.2]{Neuen22b}}]
 \label{thm:compute-isomorphisms-t-closure}
 Let $G_1,G_2$ be two graphs and let $X_1 \subseteq V(G_1)$ and $X_2 \subseteq V(G_2)$.
 Also, let $\Gamma \leq \Sym(X_1)$ be a $\mgamma_t$-group and $\theta\colon X_1 \rightarrow X_2$ a bijection.
 Moreover, let $D_i \coloneqq \cl_t^{G_i}(X_i)$ for $i \in \{1,2\}$  and define
 \[\Gamma'\theta'\coloneqq\{\varphi \in \Iso((G_1,X_1),(G_2,X_2)) \mid \varphi[X_1] \in \Gamma\theta\}[D_1].\]
 
 Then $\Gamma' \in \mgamma_t$.
 Moreover, there is an algorithm computing a $\mgamma_t$-group $\Delta \leq \Sym(D_1)$ and a bijection $\delta\colon D_1 \rightarrow D_2$ such that
 \[\Gamma'\theta'\subseteq \Delta\delta\]
 in time $n^{\CO((\log t)^{c})}$ for some absolute constant $c$ where $n\coloneqq |V(G_1)|$.
\end{theorem}

\section{Exploiting the Structure of Graphs Excluding a Minor}
\label{sec:exploit-minor}

In the following, we first give a more detailed description of the high-level strategy for building a faster isomorphism test for graph classes that exclude a fixed minor.
In particular, we state the two main technical theorems which build the groundwork for the isomorphism test.

\subsection{The Strategy}

The basic idea for our isomorphism test is to follow the decomposition framework outlined in the previous section.
Let $G_1$ and $G_2$ be two connected graphs that exclude $K_h$ as a minor (note that it is always possible to restrict to connected graphs by considering the connected components of the input graphs separately).
To apply the decomposition framework outlined in the previous section, we need to compute subsets $D_i \subseteq V(G_i)$, $i \in \{1,2\}$, such that
\begin{enumerate}[label=(\Alph*)]
 \item\label{i:inv} the subsets $D_1,D_2$ are isomorphism-invariant, i.e., each isomorphism from $G_1$ to $G_2$ maps $D_1$ to $D_2$,
 \item\label{i:small} for each connected component $Z_i$ of $G_i-D_i$ it holds $|N_{G_i}(Z_i)| < h$ and,
 \item\label{i:d} one can efficiently compute a $\mgamma_d$-group $\Delta \leq \Sym(D_1)$ and a bijection $\delta\colon D_1 \rightarrow D_2$ such that $\Iso(G_1,G_2)[D_1] \subseteq \Delta\delta$.
\end{enumerate}
In such a setting, the decomposition framework can be applied as follows.
For every pair of connected components $Z^1_{j_1}$ and $Z^2_{j_2}$ of $G_1 - D_1$ and $G_2 - D_2$, respectively, the algorithm recursively computes the set of isomorphisms from $G_1[Z^1_{j_1} \cup S^1_{j_1}]$ to $G_2[Z^2_{j_2} \cup S^2_{j_2}]$ where $S^i_{j_i} \coloneqq N_{G_i}(Z^i_{j_i})$, $i\in\{1,2\}$.
Then, the set of isomorphisms from $G_1$ to $G_2$ can be computed by combining Theorem \ref{thm:multiple-labeling-cosets-isomorphism} and \ref{thm:coset-labeled-hypergraphs-gamma-d}.
Recall that Theorem \ref{thm:multiple-labeling-cosets-isomorphism} handles the case in which $S^1_{j_1}=S^1_{j_2}$ for all connected components $Z^1_{j_1},Z^1_{j_2}$ of $G_1 - D_1$.
To achieve the desired running time for this case, we exploit Property \ref{i:small}.
For Theorem \ref{thm:coset-labeled-hypergraphs-gamma-d}, which handles the case of distinct separators $S^1_{j_1} \neq S^1_{j_2}$, we require sufficient structural information of the sets $D_1$ and $D_2$.
More precisely, we require Property \ref{i:d} to ensure the desired time bound.

\medskip

Now, we turn to the question how to find the sets $D_1$ and $D_2$ satisfying Properties \ref{i:inv}, \ref{i:small} and \ref{i:d}.
The central idea is to build on the closure operator $\cl_t^{G_i}$ (where $t$ is polynomially bounded in $h$).
We construct the sets by computing the closure $D_i \coloneqq \cl_t^{G_i}(X_i)$ for some suitable initial set $X_i$.
The first key insight is that this process of growing the sets $X_i$ can only be stopped by separators of small size which ensures Property \ref{i:small}.

\begin{theorem}
 \label{thm:small-separator-for-t-cr-bounded-closure}
 Let $G$ be a graph that excludes $K_h$ as a topological subgraph and let $X \subseteq V(G)$.
 Let $t \geq 3h^3$ and define $D \coloneqq \cl_t^{G}(X)$.
 Let $Z$ be the vertex set of a connected component of $G - D$.
 Then $|N_G(Z)| < h$.
\end{theorem}

Observe that the theorem addresses graphs that only exclude $K_h$ as a topological subgraph which is a weaker requirement than excluding $K_h$ as a minor.
A proof of this theorem, which forms the first main technical contribution of this paper, is provided in Subsection \ref{sec:small}.
As a central tool, it is argued that graphs, for which all color classes under the Color Refinement algorithm are large, contain large numbers of vertex-disjoint trees with predefined color patterns.
The vertex-disjoint trees then allow for the construction of a topological minor on the vertex set $N_G(Z)$.

\medskip

In order to ensure Property \ref{i:d}, we need sufficient structural information for the sets $D_i$, $i\in\{1,2\}$.
Using Theorem \ref{thm:compute-isomorphisms-t-closure}, we are able to extend structural information in form of a $\mgamma_d$-group from the sets $X_i$ to the supersets $D_i\supseteq X_i$, $i \in \{1,2\}$.

Hence, the main task that remains to be solved is the computation of the initial isomorphism-invariant sets $X_1$ and $X_2$ as well as suitable restrictions on the set $\Iso(G_1,G_2)[X_1] = \{\varphi[X_1] \mid \varphi \in \Iso(G_1,G_2)\}$.
Ideally, one would like to compute a $\mgamma_d$-group $\Gamma \leq \Sym(X_1)$ and a bijection $\theta\colon X_1 \rightarrow X_2$ such that $\Iso(G_1,G_2)[X_1] \subseteq \Gamma\theta$.
But this is not always possible.
For example, for a cycle $C_p$ of length $p$ where $p$ is a prime number, it is only possible to choose $X = V(C_p)$ (because $C_p$ is vertex-transitive) and $\Aut(C_p) \notin \mgamma_d$ for all $p > d$.

However, we are able to prove that there are isomorphism-invariant sets $X_1$ and $X_2$ such that, after individualizing a single vertex $v_1 \in X_1$ and $v_2 \in X_2$ in each input graph, the set $\Iso((G_1,v_1),(G_2,v_2))[X_1] = \{\varphi[X_1] \mid \varphi \in \Iso(G_1,G_2), v_1^{\varphi} = v_2\}$ has the desired structure.
This is achieved by the next theorem which forms the second main technical contribution of this paper
and again relies on the closure operator $\cl_t^{G}$.

Recall the definition of the constant $a$ from Theorem \ref{thm:average-degree-excluded-minor}.
Without loss of generality assume $a \geq 2$.

\begin{theorem}
 \label{thm:initial-color}
 Let $t \geq a^2h^3\log h$.
 There is a polynomial-time algorithm that, given a connected vertex-colored graph $G$, either correctly concludes that $G$ has a minor isomorphic to $K_h$
 or computes a pair-colored graph $(G',\chi')$ and a set $X \subseteq V(G')$ such that
 \begin{enumerate}
  \item $X = \{v \in V(G') \mid \chi'(v,v) = c\}$ for some color $c \in \{\chi'(v,v) \mid v \in V(G')\}$,
  \item\label{item:initial-color-1} $X \subseteq \cl_t^{(G',\chi')}(v)$ for every $v \in X$, and
  \item\label{item:initial-color-2} $X \subseteq V(G)$.
 \end{enumerate}
 Moreover, the output of the algorithm is isomorphism-invariant with respect to $G$.
\end{theorem}

Observe that Property \ref{item:initial-color-1} and \ref{item:initial-color-2} of the theorem imply that $(\Aut(G))_v[X] \in \mgamma_t$ for all $v \in X$ by Theorem \ref{thm:compute-isomorphisms-t-closure} (by setting $X_1 = X_2 = \{v\}$ and $\Gamma\theta$ as the singleton set containing the unique bijection from $X_1$ to $X_2$).

For technical reasons, the theorem actually provides a second graph $(G_i',\chi_i')$ for both input graphs $G_i$.
Intuitively speaking, one can think of $G_i'$ as an extension of $G_i$ which allows us to build additional structural information about $G_i$ into the graph structure of $G_i'$.
Also, the algorithm heavily exploits the $2$-dimensional Weisfeiler-Leman algorithm leading to pair-colored graphs.

The remainder of this section is devoted to proving both theorems above.
First, Theorem \ref{thm:small-separator-for-t-cr-bounded-closure} is proved in Subsection \ref{sec:small} and then Theorem \ref{thm:initial-color} is proved in Subsection \ref{sec:initial-color}.
Afterwards, the complete algorithm is assembled in Section \ref{sec:isomorphism}.

\subsection{Finding Separators of Small Size}\label{sec:small}

In this subsection, we give a proof of Theorem \ref{thm:small-separator-for-t-cr-bounded-closure}.
Let us start by giving some intuition for the proof.
Let $G$ be a graph and let $X \subseteq V(G)$.
Also define $D \coloneqq \cl_t^{G}(X)$ and let us suppose for simplicity that $G - D$ is connected with vertex set $Z \coloneqq V(G) \setminus D$.
Assume towards a contradiction that $|N_G(Z)| \geq h$ and let us fix $h$ distinct vertices $v_1,\dots,v_h \in N_G(Z)$.
We need to show that $G$ contains a topological subgraph isomorphic to $K_h$.
We use the vertices $v_1,\dots,v_h$ as the vertices of this topological subgraph, i.e., our task is to find internally vertex-disjoint paths $P_{ii'}$ connecting $v_i$ and $v_{i'}$ for all $ii' \in \binom{[h]}{2}$.
Actually, we shall achieve a stronger result by constructing vertex-disjoint, connected subgraphs $H_1,\dots,H_r$ of $G - D$, where $r \coloneqq \binom{h}{2}$, such that $v_i \in N_G(V(H_j))$ for all $i \in [h]$ and $j \in [r]$, i.e., each subgraph $H_j$ is adjacent to all the vertices $v_1,\dots,v_h$ (see Figure \ref{fig:topological-subgraph-from-large-separator}).
Note that each such subgraph $H_j$ can be used to construct one of the paths $P_{ii'}$ for $ii' \in \binom{[h]}{2}$. 

\begin{figure}
 \centering
 \begin{tikzpicture}
  \draw (0,0) ellipse (0.8cm and 1.8cm);
  \draw[rounded corners] (1.2,-1.8) rectangle (7.6,1.8);
  
  \node at (-1.2,0) {$D$};
  \node at (8.0,0) {$Z$};
  
  \node[ellipse,minimum height=30pt,minimum width=10pt,align=center,lightcolor1] (v1) at (2.5,0) {};
  \node[ellipse,minimum height=30pt,minimum width=10pt,align=center,lightcolor2] (v2) at (3.5,0) {};
  \node[ellipse,minimum height=30pt,minimum width=10pt,align=center,lightcolor3] (v3) at (4.5,0) {};
  \node[ellipse,minimum height=30pt,minimum width=10pt,align=center,lightcolor4] (v4) at (5.5,0) {};
  \node[ellipse,minimum height=30pt,minimum width=10pt,align=center,lightcolor5] (v5) at (6.5,0) {};
  \node[ellipse,minimum height=30pt,minimum width=10pt,align=center,lightcolor6] (v6) at (2,1.2) {};
  \node[ellipse,minimum height=30pt,minimum width=10pt,align=center,lightcolor7] (v7) at (3,1.2) {};
  \node[ellipse,minimum height=30pt,minimum width=10pt,align=center,lightcolor8] (v8) at (4,1.2) {};
  \node[ellipse,minimum height=30pt,minimum width=10pt,align=center,lightcolor9] (v9) at (5,1.2) {};
  \node[ellipse,minimum height=30pt,minimum width=10pt,align=center,lightcolor10] (v10) at (6,1.2) {};
  \node[ellipse,minimum height=30pt,minimum width=10pt,align=center,lightcolor11] (v11) at (2,-1.2) {};
  \node[ellipse,minimum height=30pt,minimum width=10pt,align=center,lightcolor12] (v12) at (3,-1.2) {};
  \node[ellipse,minimum height=30pt,minimum width=10pt,align=center,lightcolor13] (v13) at (4,-1.2) {};
  \node[ellipse,minimum height=30pt,minimum width=10pt,align=center,lightcolor14] (v14) at (5,-1.2) {};
  \node[ellipse,minimum height=30pt,minimum width=10pt,align=center,lightcolor15] (v15) at (6,-1.2) {};
  
  \scoped[on background layer]{
  \foreach \i/\j in {1/2,2/3,3/4,6/7,7/8,3/8,3/9,4/15,5/15,9/10,11/12,12/13,13/14,14/15,4/10,5/10,8/9}{
   \draw[thickedge] (v\i.center) edge (v\j.center);
  }
  }
    
  \node[emptyvertex,fill=dandelion!80,label = {left:$v_2$}] (w1) at (0.4,0) {};
  \node[emptyvertex,fill=purple!80,label = {left:$v_1$}] (w2) at (0.2,1.2) {};
  \node[emptyvertex,fill=darkspringgreen!80,label = {left:$v_3$}] (w3) at (0.2,-1.2) {};
  
  \foreach \s in {-1,0,1}{
   \node[tinyvertex,color1] (t1-\s) at ($(2.5,0)+(0,0.3*\s)$) {};
   \draw[thick] (w1) edge (t1-\s);
   \node[tinyvertex,color6] (t6-\s) at ($(2,1.2)+(0,0.3*\s)$) {};
   \draw[thick] (w2) edge (t6-\s);
   \node[tinyvertex,color4] (t11-\s) at ($(2,-1.2)+(0,0.3*\s)$) {};
   \draw[thick] (w3) edge (t11-\s);
  }
  
  \foreach \i in {2,3,4,7,8,12,13,14,15}{
   \foreach \s in {-1,0,1}{
    \node[tinyvertex,color\i] (t\i-\s) at ($(v\i.center)+(0,0.3*\s)$) {};
   }
  }
  
  \foreach \i/\j in {1/2,2/3,3/4,6/7,7/8,3/8,11/12,12/13,13/14,14/15,4/15}{
   \foreach \s in {-1,0,1}{
    \draw[thick] (t\i-\s) edge (t\j-\s);
   }
  }
 \end{tikzpicture}
 \caption{Visualization of the construction of a topological subgraph $K_h$ for $h = 3$.
  The figure shows $\binom{h}{2} = 3$ vertex-disjoint subgraphs $H_1,H_2,H_3$ of $G[Z]$ each of which is adjacent to $v_1,v_2,v_3$.
  Note that only the vertices and edges appearing in one the graphs $H_j$ are shown.
  Indeed, by assumption, all color classes in $Z$ contain at least $3h^{3} = 81$ vertices.
  The edges of the color class graph on $Z$ are visualized by thick, gray connections.}
 \label{fig:topological-subgraph-from-large-separator}
\end{figure}

To construct the subgraphs $H_j$, we build on the assumption that $D = \cl_t^{G}(X)$ which, by Definition \ref{def:t-cr-bounded}, means that there is a vertex-coloring $\chi$ of $G$ that is stable with respect to the Color Refinement algorithm, $|[v]_{\chi}| = 1$ for all $v \in D$, and $|[w]_{\chi}| > t \geq 3h^{3}$ for all $w \in Z$.
Towards this end, consider the \emph{color class graph of $(G,\chi)$ on $Z$}, denoted by $G[[\chi,Z]]$, which is the graph with vertex set $V(G[[\chi,Z]]) \coloneqq \chi(Z)$ and edge set $E(G[[\chi,Z]]) \coloneqq \{\chi(w_1)\chi(w_2) \mid w_1w_2 \in E(G[Z])\}$.
Clearly, $G[[\chi,Z]]$ is connected since $G[Z]$ is connected.
Since $v_1,\dots,v_h \in N_G(Z)$, there are colors $c_1,\dots,c_h$ such that $v_i \in N_G(\chi^{-1}(c_i))$ for every $i \in [h]$, i.e., $v_i$ is adjacent to some vertex contained in the color class corresponding to color $c_i$.
Since $\chi$ is stable with respect to the Color Refinement algorithm and $|[v_i]_{\chi}| = 1$ by assumption, it actually follows that $N_G(v_i) \subseteq \chi^{-1}(c_i)$, i.e, $v_i$ is adjacent to every vertex from the color class corresponding to $c_i$.

Now, let $T$ be a Steiner tree for $c_1,\dots,c_h$ in $G[[\chi,Z]]$, i.e., $T$ is a subtree of $G[[\chi,Z]]$ with leaves exactly $c_1,\dots,c_h$ (such a tree clearly always exists).
We construct the graphs $H_1,\dots,H_r$ in such a way that each $H_j$ ``mimics'' the structure of $T$, i.e., for every color $c$ that is a vertex of $T$, the graph $H_j$ contains exactly one vertex with color $c$, and two vertices with colors $c,c'$ are adjacent if and only if $cc'$ is an edge of $T$ (see Figure \ref{fig:topological-subgraph-from-large-separator}).
Here, we crucially use that all color classes are sufficiently large which gives us enough ``capacity'' to fit all the subgraphs $H_j$ into the desired color classes at the same time.

Before turning to the formal proof of Theorem \ref{thm:small-separator-for-t-cr-bounded-closure}, let us briefly comment on the assumption that $G - D$ is connected which we made for the above argument.
If $G - D$ is not connected, it may happen that color classes are not sufficiently large when restricted to a connected component $Z$, since color classes may span over multiple connected components of $G - D$.
However, this is not a problem since we can consider all these connected components together.
Indeed, for our argument, we do not require that $G[Z]$ is connected, but it suffices that $G[[\chi,Z]]$ is connected.

The following lemma reformulates the task of proving Theorem \ref{thm:small-separator-for-t-cr-bounded-closure} as indicated above.

Recall that for a vertex-colored graph $(G,\chi)$ and $W \subseteq V(G)$, we define the \emph{color class graph of $(G,\chi)$ on $W$} to be the graph $G[[\chi,W]]$ with vertex set $V(G[[\chi,W]]) \coloneqq \chi(W)$ and edge set $E(G[[\chi,W]]) \coloneqq \{\chi(w_1)\chi(w_2) \mid w_1w_2 \in E(G[W])\}$.

\begin{lemma}
 \label{la:small-color-classes-excluded-kh}
 Let $h \geq 1$.
 Let $G$ be a graph and $V(G) = V_1 \uplus V_2$ be a partition of the vertex set of $G$.
 Also let $\chi$ be a vertex-coloring of $G$ and suppose that
 \begin{enumerate}
  \item $G[[\chi,V_2]]$ is connected,
  \item $|V_1| \geq h$ and $N_G(V_2) = V_1$,
  \item $|[v]_\chi| = 1$ for all $v \in V_1$,
  \item $|[w]_\chi| \geq 3h^3$ for all $w \in V_2$, and
  \item $\chi$ is stable with respect to the Color Refinement algorithm.
 \end{enumerate}
 Then $G$ has a topological subgraph isomorphic to $K_h$.
\end{lemma}

Before showing the lemma, we give a proof for Theorem \ref{thm:small-separator-for-t-cr-bounded-closure} based on Lemma \ref{la:small-color-classes-excluded-kh}.

\begin{proof}[Proof of Theorem \ref{thm:small-separator-for-t-cr-bounded-closure}]
 Let $\chi$ be the final vertex-coloring that is stable under the $t$-CR-bounded algorithm with respect to the initial set $X$.
 Let $Z$ be a connected component of $G - D$ and assume for sake of contradiction that $|N_G(Z)| \geq h$.
 Let $V_2 \coloneqq \{v \in V(G) \mid \chi(v) \in \chi(Z)\}$ and $V_1 \coloneqq N_G(V_2)$ and define $H \coloneqq G[V_1 \cup V_2]$.
 We have $|V_1| \geq |N_G(Z)| \geq h$. 
 Also, $|[v]_\chi|=1$ for all $v \in V_1 \subseteq D = \cl_t^{G}(X)$.
 Moreover, $\chi|_H$ is stable under the Color Refinement algorithm for the graph $H$ and $H[[\chi|_H,V_2]]$ is connected since $G[Z]$ is connected.
 Finally, $|[w]_\chi| > t \geq 3h^3$ for all $w \in V_2$ since $\chi$ is stable under the $t$-CR-bounded algorithm and $V_2 \cap D = \emptyset$.
 So by Lemma \ref{la:small-color-classes-excluded-kh} the graph $H$ has a topological subgraph isomorphic to $K_h$.
 But $H$ is a subgraph of $G$ and hence, $G$ also has a topological subgraph isomorphic to $K_h$.
 This gives a contradiction.
\end{proof}

We now turn to proving Lemma \ref{la:small-color-classes-excluded-kh} where we aim to construct a topological subgraph isomorphic to $K_h$.
The vertices of the topological subgraph are located in the set $V_1$.
This leaves the task to construct disjoint paths between vertices from $V_1$ using the vertices from the set $V_2$.
Actually, as already described above, it turns out to be more convenient to construct a large number of disjoint trees each of which can be used to obtain a single path connecting two vertices in $V_1$.

Let $G$ be a graph, let $\chi\colon V(G)\to C$ be a vertex-coloring and let $T$ be a tree with vertex set $V(T)=C$.
A subgraph $H\subseteq G$ \emph{agrees} with $T$ if $\chi|_{V(H)} \colon H \cong T$, i.e., the coloring $\chi$ induces an isomorphism between $H$ and $T$.
Equivalently, $H$ agrees with $T$ if $|V(H)\cap\chi^{-1}(c)|=1$ for every $c \in C$ and $c_1c_2\in E(T)$ if and only if $H[\chi^{-1}(c_1),\chi^{-1}(c_2)]$ contains an edge for all $c_1,c_2 \in C$.
Observe that each $H\subseteq G$ that agrees with a tree $T$ is also a tree.
Let $H_1,\ldots,H_k \subseteq G$ be $k$ pairwise vertex-disjoint subgraphs that agree with the tree $T$.
For a color $c\in C$ we define $V_c(H_1,\dots,H_k)\coloneqq\bigcup_{i=1}^k V(H_i)\cap\chi^{-1}(c)$.
Since $H_1,\ldots,H_k$ are pairwise vertex-disjoint, it holds that $|V_c(H_1,\ldots,H_k)|=k$.
The \emph{extension set} for $H_1,\dots,H_k$ and a color $c \in C$ is defined as
\begin{align*}
 W_c \coloneqq W_c(H_1,\ldots H_k) \coloneqq \Big\{ v \in \chi^{-1}(c) \;\Big|\; &\text{there are $k+1$ pairwise vertex-disjoint}\\
                                                                                 &\text{connected graphs $H_1',\ldots,H_{k+1}'$}\\
                                                                                 &\text{that agree with $T$ such that}\\
                                                                                 &\text{$V_c(H_1',\ldots,H_{k+1}')=V_c(H_1,\ldots,H_k)\cup\{v\}$}\Big\}.
\end{align*}
A visualization is also given in Figure \ref{fig:tree-extension-set}.

\begin{figure}
 \centering
 \begin{tikzpicture}
  \foreach \x/\y in {0/0,6/0,3/2,3/4}{
   \draw (\x,\y) ellipse (2.4cm and 0.6cm);
   \node[midvertex] (\x\y1) at (\x - 1.8,\y) {};
   \node[midvertex] (\x\y2) at (\x - 0.9,\y) {};
   \node[midvertex] (\x\y3) at (\x + 0.0,\y) {};
   \node[midvertex] (\x\y4) at (\x + 0.9,\y) {};
   \node[midvertex] (\x\y5) at (\x + 1.8,\y) {};
  }
  \node at (1.2,4.2) {{\footnotesize $v_1$}};
  \node at (2.1,4.2) {{\footnotesize $v_2$}};
  \node at (3.0,4.2) {{\footnotesize $v_3$}};
  \node at (3.9,4.2) {{\footnotesize $v_4$}};
  \node at (4.8,4.2) {{\footnotesize $v_5$}};
  
  \node at (6.0,4.2) {{\footnotesize $\chi^{-1}(c_1)$}};
  \node at (6.0,2.2) {{\footnotesize $\chi^{-1}(c_2)$}};
  \node at (-2.0,-0.8) {{\footnotesize $\chi^{-1}(c_3)$}};
  \node at (8.0,-0.8) {{\footnotesize $\chi^{-1}(c_4)$}};
  
  \foreach \v/\w in {00/32,60/32,32/34}{
   \draw[thick] (\v1) edge (\w1);
   \draw[thick] (\v2) edge (\w2);
   \draw[thick] (\v3) edge (\w3);
   \draw[thick] (\v4) edge (\w4);
   \draw[thick] (\v5) edge (\w5);
  }
  \draw[thick] (322) edge (341);
  \draw[thick] (323) edge (342);
  \draw[thick] (324) edge (343);
  \draw[thick] (325) edge (344);
  \draw[thick] (321) edge (345);
  
  \scoped[on background layer]{
   \draw[line width = 6pt, blue!40] (001.center) edge (321.center);
   \draw[line width = 6pt, blue!40] (601.center) edge (321.center);
   \draw[line width = 6pt, blue!40] (321.center) edge (341.center);
   
   \draw[line width = 6pt, red!40] (002.center) edge (322.center);
   \draw[line width = 6pt, red!40] (602.center) edge (322.center);
   \draw[line width = 6pt, red!40] (322.center) edge (342.center);
   
   \draw[line width = 6pt, darkpastelgreen!40] (005.center) edge (325.center);
   \draw[line width = 6pt, darkpastelgreen!40] (605.center) edge (325.center);
   \draw[line width = 6pt, darkpastelgreen!40] (325.center) edge (344.center);
   
   \draw[line width = 6pt, darkpastelgreen!40] (005.center) edge (325.center);
   \draw[line width = 6pt, darkpastelgreen!40] (605.center) edge (325.center);
  }
 \end{tikzpicture}
 \caption{The figure shows three vertex-disjoint subgraphs $H_1,H_2,H_3$ (in red, blue and green) that agree with the tree $T$ with node set $V(T) = \{c_1,c_2,c_3,c_4\}$ and edge set $E(T) = \{c_1c_2,c_2c_3,c_2c_4\}$.
  The extension set for the color $c_1$ is $W_{c_1}(H_1,H_2,H_3) = \{v_3,v_5\}$.
  Note that there is a subgraph $H_4$ that agrees with $T$ containing $v_3$ and being vertex-disjoint from $H_1,H_2,H_3$.
  However, to argue that $v_5$ is contained in the extension set, one also has to modify at least one of the already constructed subgraphs $H_1,H_2,H_3$.}
 \label{fig:tree-extension-set}
\end{figure}

Intuitively speaking, the idea is to construct subgraphs $H_1,\dots,H_r$ that agree with $T$ one after the other.
Suppose we already constructed subgraphs $H_1,\dots,H_k$ for some $k < r$, and we aim to construct the next subgraph $H_{k+1}$.
To do so, we can fix an arbitrary root of $T$ and construct $H_{k+1}$ in a bottom-up fashion starting at the leaves of $T$.
Unfortunately, when constructing $H_{k+1}$ this way, we may get stuck (i.e., it is not possible to extend the partial subgraph $H_{k+1}$ any further without using vertices from $H_1,\dots,H_k$) which forces us to make changes to the already constructed $H_1,\dots,H_k$.
We formalize this idea by providing, for each node $c$ of $T$, a lower bound on the size of the extension set when restricting to the subtree of $T$ rooted at $c$.
In the end, we obtain that the extension set for the root contains at least one element and hence, there exist $k+1$ disjoint connected subgraphs $H_1',\ldots,H_{k+1}'$ (which may be completely different from $H_1,\dots,H_k$, but this is not a problem for our purposes).

The next lemma serves as an important intermediate step to deal with long induced paths of $T$.
Indeed, this case is critical since the number of nodes of $T$ of degree $2$ may be unbounded.
To be more precise, in our application, the number of leaves of $T$ is bounded by $h$, and thus, the number of internal nodes of degree at least $3$ is also bounded by $h$.
All color classes in $V_2$ have size at least $3h^{3}$ (which is much larger than $r = \binom{h}{2}$) which means that, when providing lower bounds on the size of the extension sets, we can afford some small loss at all internal nodes of degree at least $3$.

In contrast, the number of nodes of degree $2$ in $T$ is unbounded, and hence we cannot afford any loss at those nodes when lower-bounding the size of an extension set.
We use the following lemma to deal with long induced paths of $T$.

\begin{lemma}
\label{la:find-disjoint-paths}
 Let $G$ be a graph and let $\chi\colon V(G)\to C$ be a vertex-coloring and let $P$ be a path with vertex set
 $V(P)=C=\{c_1,\ldots,c_s\}$ and edge set $E(P)=\{c_ic_{i+1}\mid i\in[s-1]\}$.
 Also, suppose that $G[\chi^{-1}(c_i),\chi^{-1}(c_{i+1})]$ is a
 non-empty biregular graph for every $i\in[s-1]$.
 Let $m \coloneqq \min_{c\in C} |\chi^{-1}(c)|$.
 Let $H_1,\ldots,H_k$ be $k$ pairwise vertex-disjoint path graphs that agree with $P$.
 Let $X\subseteq\chi^{-1}(c_1)\setminus V_{c_1}(H_1,\ldots,H_k)$ and
 \begin{align*}
  W_{X,c_s} \coloneqq \Big\{ v \in \chi^{-1}(c_s) \;\Big|\; &\text{there are $k+1$ pairwise vertex-disjoint}\\
                                                            &\text{path graphs $H_1',\ldots,H_{k+1}'$ such that}\\
                                                            &\text{$V_{c_1}(H_1',\ldots,H_{k+1}')=V_{c_1}(H_1,\ldots,H_k)\cup\{x\}$ for some $x\in X$}\\
                                                            &\text{and $V_{c_s}(H_1',\ldots,H_{k+1}')=V_{c_s}(H_1,\ldots,H_k)\cup\{v\}$}\Big\}.
 \end{align*}
 Then
 \[\frac{|W_{X,c_s}|}{|\chi^{-1}(c_s)|} \geq \frac{|X|}{|\chi^{-1}(c_1)|} - \frac{k}{m}.\]
\end{lemma}

To grasp the meaning of this lemma, it is helpful to think of $P$ as an induced path of $T$ with $c_1$ being a descendant of $c_s$ (see also Figure \ref{fig:extension-set-induced-paths}).
The set $X$ is the extension set with respect to the subtree rooted at $c_1$, and we aim to provide a lower bound for the size of the extension set with respect to the subtree rooted at $c_s$.
Towards this end, the lemma provides a lower bound on the size of $W_{X,c_s}$ which is a subset of the desired extension set.

\begin{figure}
 \centering
 \begin{tikzpicture}
  \foreach[count=\a] \x/\y in {0/0,3/0,1.5/1,1.5/2.5,1.5/4}{
   \draw (\x,\y) ellipse (1.2cm and 0.3cm);
   \node[tinyvertex] (\a1) at (\x - 0.9,\y) {};
   \node[tinyvertex] (\a2) at (\x - 0.6,\y) {};
   \node[tinyvertex] (\a3) at (\x - 0.3,\y) {};
  }
  
  \node at (3.3,1.1) {{\footnotesize $\chi^{-1}(c_1)$}};
  \node at (3.3,2.6) {{\footnotesize $\chi^{-1}(c_i)$}};
  \node at (3.3,4.1) {{\footnotesize $\chi^{-1}(c_s)$}};
  
  \foreach \i in {1,2,3}{
   \draw[thick] (1\i) edge (3\i);
   \draw[thick] (2\i) edge (3\i);
   \draw[thick,decorate,decoration=snake,segment length = 0.5cm] (3\i) -- (4\i);
   \draw[thick,decorate,decoration=snake,segment length = 0.5cm] (4\i) -- (5\i);
  }
  
  \node at (1.5,1.8) {{\Large \vdots}};
  \node at (1.5,3.3) {{\Large \vdots}};
  
  \node at (0,-0.7) {{\Large \vdots}};
  \node at (3,-0.7) {{\Large \vdots}};
  
  \draw[gray!60, fill = gray!60] (2,1) ellipse (0.4cm and 0.18cm);

  \foreach[count=\a] \x/\y in {7/0,10/0,8.5/1,8.5/2.5,8.5/4}{
   \draw (\x,\y) ellipse (1.2cm and 0.3cm);
   \node[tinyvertex] (\a1) at (\x - 0.9,\y) {};
   \node[tinyvertex] (\a2) at (\x - 0.6,\y) {};
   \node[tinyvertex] (\a3) at (\x - 0.3,\y) {};
  }
  
  \node at (10.3,1.1) {{\footnotesize $\chi^{-1}(c_1)$}};
  \node at (10.3,2.6) {{\footnotesize $\chi^{-1}(c_i)$}};
  \node at (10.3,4.1) {{\footnotesize $\chi^{-1}(c_s)$}};
  
  \node at (9.5,1.8) {{\Large \vdots}};
  \node at (9.5,3.3) {{\Large \vdots}};
  
  \node at (7,-0.7) {{\Large \vdots}};
  \node at (10,-0.7) {{\Large \vdots}};
  
  \draw[gray!60, fill = gray!60] (9,1) ellipse (0.4cm and 0.18cm);
  
  \node[tinyvertex] (14) at (7.5,0) {};
  \node[tinyvertex] (24) at (10.5,0) {};
  \node[tinyvertex] (34) at (9,1) {};
  \node[tinyvertex] (44) at (8.8,2.5) {};
  \node[tinyvertex] (45) at (9.2,2.5) {};
  \node[tinyvertex] (54) at (9,4) {};
  
  \foreach \i in {1,2,3,4}{
   \draw[thick] (1\i) edge (3\i);
   \draw[thick] (2\i) edge (3\i);
  }
  
  \node at (9.2,4) {{\footnotesize $v$}};
  
  \draw[thick,decorate,decoration=snake,segment length = 0.5cm] (31) -- (42);
  \draw[thick,decorate,decoration=snake,segment length = 0.5cm] (32) -- (44);
  \draw[thick,decorate,decoration=snake,segment length = 0.5cm] (33) -- (45);
  \draw[thick,decorate,decoration=snake,segment length = 0.5cm] (34) -- (43);
  
  \draw[thick,decorate,decoration=snake,segment length = 0.5cm] (42) -- (51);
  \draw[thick,decorate,decoration=snake,segment length = 0.5cm] (43) -- (52);
  \draw[thick,decorate,decoration=snake,segment length = 0.5cm] (44) -- (53);
  \draw[thick,decorate,decoration=snake,segment length = 0.5cm] (45) -- (54);
 \end{tikzpicture}
 \caption{The left side shows part of a tree $T$ rooted at $c_s$ with a long induced path between $c_s$ and $c_1$.
  We already constructed three subgraphs $H_1,H_2,H_3$ that agree with $T$.
  The extension set $X$ for the subtree rooted at $c_1$ is shown in gray.
  We have $v \in W_{X,c_s}$ via witnessing paths shown on the right side.
  In particular, $v$ is contained in the extension set $W_{c_s}$ for the color $c_s$.}
 \label{fig:extension-set-induced-paths}
\end{figure}

\begin{proof}
 The proof uses an alternating-paths argument.
 We split the edges of $G$ into \emph{forward} and \emph{backward} edges and direct all edges accordingly. 
 Let
 \begin{align*}
  E_{\fw} \coloneqq \Big\{(v,w) \;\Big|\; &vw\in E(G)\setminus (\bigcup_{j\in[k]} E(H_j)),v\in\chi^{-1}(c_i)\text{ and }\\
                                          &w\in\chi^{-1}(c_{i+1})\text{ for some }i\in[s-1]\Big\}.
 \end{align*}
 and
 \begin{align*}
  E_{\bw} \coloneqq \Big\{(v,w) \;\Big|\; &vw\in \bigcup_{j\in[k]} E(H_j),v\in\chi^{-1}(c_{i+1})\text{ and }\\
                                          &w\in\chi^{-1}(c_i)\text{ for some }i\in[s-1]\Big\}.
 \end{align*}
 We consider directed paths that start in $X\subseteq \chi^{-1}(c_1)\setminus V_{c_1}(H_1,\ldots,H_k)$.
 A path $v_1,\ldots,v_t$ in $G$ is \emph{admissible} if
 \begin{enumerate}
  \item $v_1 \in X$,
  \item $(v_i,v_{i+1}) \in E_{\fw} \cup E_{\bw}$, and
  \item if $(v_i,v_{i+1})\in E_{\fw}$ and $v_{i+1} \in \bigcup_{j\in[k]} V(H_j)$ then $i\leq t-2$ and $(v_{i+1},v_{i+2}) \in E_{\bw}$
 \end{enumerate}
 for all $i \in [t-1]$. Let
 \[A\coloneqq\{v\in V(G) \mid \text{there is an admissible path $v_1,\ldots,v_t$ such that } v_t=v\}.\]
 For a color $c\in C$ let $A_c\coloneqq A \cap \chi^{-1}(c)$.
 
 \begin{claim}
  \label{cl:extension-set-contains-admissible}
  $A_{c_s} \subseteq W_{X,c_s}$.
 \end{claim}
 \begin{claimproof}
  Let $v \in A_{c_s}$ and let $v_1,\ldots,v_t$ be an admissible path of minimal length $t$ such that $v = v_t$.
  Let $H_{k+1}$ be the corresponding path graph with $V(H_{k+1}) \coloneqq \{v_1,\ldots,v_t\}$ and $E(H_{k+1}) \coloneqq \{v_iv_{i+1}\mid i\in[t-1]\}$.
  Consider the graph $H$ with vertex set $V(H) \coloneqq \bigcup_{i\in[k+1]} V(H_i)$ and edge set
  \begin{equation*}
   E(H) \coloneqq \left(\bigcup_{i\in[k]} E(H_i) \setminus E(H_{k+1})\right) \cup \left(E(H_{k+1})\setminus\bigcup_{i\in[k]} E(H_i)\right).
  \end{equation*}
  We claim that $H$ is the disjoint union of $(k+1)$ many graphs $H_1',\ldots,H_{k+1}'$ (and possibly isolated vertices) that agree with $P$.
  
  To see this, consider some vertex $u \in V(H)$ such that $u \in \chi^{-1}(c_i)$ for some $i \in \{2,\dots,s-1\}$.
  If $u$ is contained in exactly one of the sets $V(H_i)$, $i \in [k+1]$, then it is easy to see that $u$ has exactly two neighbors in $H$, one in $\chi^{-1}(c_{i-1})$ and another in $\chi^{-1}(c_{i+1})$.
  Otherwise, $u \in V(H_{k+1})$ and $u \in V(H_i)$ for some $i \in [k]$ (recall that $V(H_i) \cap V(H_j) = \emptyset$ for all distinct $i,j \in [k]$).
  Then, on the admissible path $v_1,\ldots,v_t$, the vertex $u$ is either incident to two backward edges, or incident to one forward edge and one backward edge.
  In the former case, $u$ is isolated in $H$ (see Figure \ref{fig:backward-edges-bw-bw}), and in the latter case $u$ has again exactly two neighbors in $H$, one in $\chi^{-1}(c_{i-1})$ and another in $\chi^{-1}(c_{i+1})$ (see Figure \ref{fig:backward-edges-fw-bw}).
  
  \begin{figure}
   \begin{subfigure}{.48\linewidth}
    \centering
    \begin{tikzpicture}
     \foreach \a in {0,1,2}{
      \draw (0,\a) ellipse (1.6cm and 0.3cm);
      \node[midvertex] (\a1) at (-1.0,\a) {};
      \node[midvertex] (\a2) at (-0.4,\a) {};
      \node[midvertex] (\a3) at (0.2,\a) {};
     }
     
     \node at (-0.02,1) {{\footnotesize $u$}};
     
     \node at (2.36,0.1) {{\footnotesize $\chi^{-1}(c_{i+1})$}};
     \node at (2.2,1.1) {{\footnotesize $\chi^{-1}(c_{i})$}};
     \node at (2.36,2.1) {{\footnotesize $\chi^{-1}(c_{i-1})$}};
     
     \foreach \i in {1,2,3}{
      \draw[thick] (0\i) edge (1\i);
      \draw[thick] (1\i) edge (2\i);
     }
     
     \scoped[on background layer]{
      \draw[line width = 6pt, blue!40] (01.center) edge (11.center);
      \draw[line width = 6pt, blue!40] (11.center) edge (21.center);
      
      \draw[line width = 6pt, red!40] (02.center) edge (12.center);
      \draw[line width = 6pt, red!40] (12.center) edge (22.center);
      
      \draw[line width = 6pt, darkpastelgreen!40] (03.center) edge (13.center);
      \draw[line width = 6pt, darkpastelgreen!40] (13.center) edge (23.center);
     }
     
     \draw[line width = 2pt, gray, ->, bend right = 40] (03) edge (13);
     \draw[line width = 2pt, gray, ->, bend right = 40] (13) edge (23);
    \end{tikzpicture}
    \caption{The vertex $u$ is incident to two backward edges. It is isolated in $H$ since both backward edges are deleted.}
    \label{fig:backward-edges-bw-bw}
   \end{subfigure}
   \hfill
   \begin{subfigure}{.48\linewidth}
    \centering
    \begin{tikzpicture}
     \foreach \a in {0,1,2}{
      \draw (0,\a) ellipse (1.6cm and 0.3cm);
      \node[midvertex] (\a1) at (-1.0,\a) {};
      \node[midvertex] (\a2) at (-0.4,\a) {};
      \node[midvertex] (\a3) at (0.2,\a) {};
     }
     
     \node at (-0.02,1) {{\footnotesize $u$}};
     \node at (0.78,2) {{\footnotesize $u'$}};
     \node at (0.48,0) {{\footnotesize $u''$}};
     
     \node at (2.36,0.1) {{\footnotesize $\chi^{-1}(c_{i+1})$}};
     \node at (2.2,1.1) {{\footnotesize $\chi^{-1}(c_{i})$}};
     \node at (2.36,2.1) {{\footnotesize $\chi^{-1}(c_{i-1})$}};
     
     \foreach \i in {1,2,3}{
      \draw[thick] (0\i) edge (1\i);
      \draw[thick] (1\i) edge (2\i);
     }
     
     \node[midvertex] (24) at (1.0,2) {};
     \draw[thick] (24) edge (13);
     
     \scoped[on background layer]{
      \draw[line width = 6pt, blue!40] (01.center) edge (11.center);
      \draw[line width = 6pt, blue!40] (11.center) edge (21.center);
      
      \draw[line width = 6pt, red!40] (02.center) edge (12.center);
      \draw[line width = 6pt, red!40] (12.center) edge (22.center);
      
      \draw[line width = 6pt, darkpastelgreen!40] (03.center) edge (13.center);
      \draw[line width = 6pt, darkpastelgreen!40] (13.center) edge (23.center);
     }
      
     \draw[line width = 2pt, gray, ->, bend left = 40] (24) edge (13);
     \draw[line width = 2pt, gray, ->, bend right = 40] (13) edge (23);
    \end{tikzpicture}
    \caption{The vertex $u$ is incident to one forward and one backward edge. It is adjacent to $u'$ and $u''$ in the graph $H$.}
    \label{fig:backward-edges-fw-bw}
   \end{subfigure}
   \caption{Visualization for the construction of the graph $H$ in Claim \ref{cl:extension-set-contains-admissible} for $k = 3$.
    The edges of the graphs $H_1,H_2,H_3$ are highlighted in blue, red and green.
    The directed edges of the admissible path are shown in gray.}
   \label{fig:backward-edges}
  \end{figure}
  
  Similarly, it follows that every $u \in V(H) \cap \chi^{-1}(c_1)$ has exactly one neighbor in $\chi^{-1}(c_2)$, and every $u \in V(H) \cap \chi^{-1}(c_s)$ has exactly one neighbor in $\chi^{-1}(c_{s-1})$.
  Overall, it follows that $H$ is the disjoint union of $(k+1)$ many graphs $H_1',\ldots,H_{k+1}'$ (and possibly isolated vertices) that agree with $P$.
  In particular, $v \in W_{X,c_s}$.
 \end{claimproof}
 
 By the claim, it suffices to provide a lower bound on the size of the set $A_{c_s}$.
 Towards this end, we analyze the structure of the set $A$.
 For $j \in [k]$ let $w_{i,j}$ be the unique vertex in the set $V(H_j) \cap \chi^{-1}(c_i)$, $i \in [s]$.
 First observe that, if $w_{i,j} \in A$ then also $w_{i',j} \in A$ for all $i' < i$ since
 all vertices $w_{i',j}$ are reachable with backward edges in $E_{\bw}$.
 
 We call a vertex $b \in \bigcup_{j \in [k]} V(H_j) \setminus A$ a \emph{blocking vertex} if
 there is a vertex $v \in \bigcup_{j \in [k]} V(H_j) \cap A$ such that $(b,v) \in E_{\bw}$.
 In other words, the vertex $w_{i,j}$ is a blocking vertex if $w_{i,j} \notin A$ and $w_{i-1,j} \in A$
 (and therefore $w_{i',j} \in A$ for all $i'<i$).
 Let $B$ be the set of blocking vertices.
 By the above observation, $|B \cap V(H_j)| \leq 1$ for all  $j \in [k]$ (and $|B \cap V(H_j)| = 0$ if $V(H_j) \cap A = \emptyset$).
 Hence, $|B| \leq k$.
 Let $k_i \coloneqq |B \cap \chi^{-1}(c_i)|$ be the number of blocking vertices of color $c_i$, $i \in [s]$.
 
 \begin{claim}
  \label{cl:blocking-vertices}
  $\displaystyle\frac{|A_{c_i}|}{|\chi^{-1}(c_i)|} \geq \frac{|A_{c_1}|}{|\chi^{-1}(c_1)|} - \frac{k_1+\ldots+k_i}{m}$ for all $i \in [s]$.
 \end{claim}
 \begin{claimproof}
  The claim is proved by induction on $i \in [s]$.
  In the base case $i=1$, it holds that $k_1 = 0$ which implies the statement.
  
  For the inductive step assume that $i \geq 1$.
  Since $G[\chi^{-1}(c_i),\chi^{-1}(c_{i+1})]$ is a non-empty biregular graph, for each subset $S \subseteq \chi^{-1}(c_i)$
  it holds that
  $\frac{|N(S)\cap\chi^{-1}(c_{i+1})|}{|\chi^{-1}(c_{i+1})|}\geq\frac{|S|}{|\chi^{-1}(c_i)|}$
  as argued in the preliminaries.
  We first argue that $(N(A_{c_i})\cap\chi^{-1}(c_{i+1}))\subseteq A_{c_{i+1}}\cup B$.
  Let $v\in A_{c_i}$ and $w\in N(A_{c_i})\cap\chi^{-1}(c_{i+1})$.
  If $w \in \bigcup_{j\in[k]}V(H_j)$ then $w \in B$ or $w \in A$.
  Otherwise $w \in V(G)\setminus\bigcup_{j\in[k]}V(H_j)$ and $(v,w) \in E_{\fw}$ which means $w \in A$.
  This shows the inclusion and therefore
  \[\frac{|A_{c_{i+1}}|}{|\chi^{-1}(c_{i+1})|}\geq
  \frac{|N(A_{c_i})\cap\chi^{-1}(c_{i+1})|-k_{i+1}}{|\chi^{-1}(c_{i+1})|}
  \geq \frac{|A_{c_i}|}{|\chi^{-1}(c_i)|}-\frac{k_{i+1}}{m}.\]
  By the induction hypothesis, $\frac{|A_{c_i}|}{|\chi^{-1}(c_i)|}\geq\frac{|A_{c_1}|}{|\chi^{-1}(c_1)|}-\frac{k_1+\ldots+k_i}{m}$.
  In combination this means
  \[\frac{|A_{c_{i+1}}|}{|\chi^{-1}(c_{i+1})|}\geq
  \frac{|A_{c_1}|}{|\chi^{-1}(c_1)|}-\frac{k_1+\ldots+k_i+k_{i+1}}{m}.\qedhere\]
 \end{claimproof}

 Now, we can prove the lemma.
 We already observed in Claim \ref{cl:extension-set-contains-admissible} that $A_{c_s}\subseteq W_{X,c_s}$.
 Moreover, it holds that $X \subseteq A_{c_1}$.
 Combining this with Claim \ref{cl:blocking-vertices}, we obtain
 \begin{align*}
  \frac{|W_{X,c_s}|}{|\chi^{-1}(c_s)|} &\geq \frac{|A_{c_s}|}{|\chi^{-1}(c_s)|}\\
                                       &\geq \frac{|A_{c_1}|}{|\chi^{-1}(c_1)|}-\frac{k}{m}\\
                                       &\geq \frac{|X|}{|\chi^{-1}(c_1)|}-\frac{k}{m}.\qedhere
 \end{align*}
\end{proof}

Let $T$ be a tree.
We define $V_{\leq i}(T) \coloneqq \{t \in V(T) \mid \deg(t) \leq i\}$ and $V_{\geq i}(T) \coloneqq \{t \in V(T) \mid \deg(t) \geq i\}$.
It is well known that $|V_{\geq 3}(T)|\leq|V_{\leq 1}(T)|$.

We are now ready to prove the main technical lemma of this section which provides the desired subgraphs $H_1,\dots,H_r$.

\begin{lemma}
 \label{la:find-disjoint-trees}
 Let $G$ be a graph, let $\chi\colon V(G)\to C$ be a vertex-coloring and let $T$ be a tree with vertex set $V(T) = C$.
 Assume that $G[\chi^{-1}(c_1),\chi^{-1}(c_2)]$ is a non-empty biregular graph for every $c_1c_2\in E(T)$.
 Let $m \coloneqq \min_{c\in C} |\chi^{-1}(c)|$ and let $\ell \coloneqq 2|V_{\leq 1}(T)|+|V_{\geq 3}(T)|$.
 
 Then there are (at least) $\lfloor\frac{m}{\ell}\rfloor$ pairwise
 vertex-disjoint trees in $G$ that agree with $T$.
\end{lemma}

\begin{proof}
 We show by induction on $k$ that there are at least $k$ pairwise vertex-disjoint trees that agree with $T$ for all $k \leq \lfloor\frac{m}{\ell}\rfloor$.
 The base case $k = 0$ is trivial.
 For the inductive step assume there are $k < \lfloor\frac{m}{\ell}\rfloor$ pairwise vertex-disjoint trees $H_1,\ldots,H_k\subseteq G$ that agree with $T$.
 From $H_1,\dots,H_k$ we construct vertex-disjoint trees $H_1',\dots,H_k',H_{k+1}' \subseteq G$ that agree with $T$.
 
 Towards this end, it suffices to show that there is a color $c \in C$ such that the extension set $W_c = W_c(H_1,\dots,H_k)$ is non-empty.
 In fact, we show by induction on $n \coloneqq |V(T)|$ that for all colors $c \in C$ it holds that $\frac{|W_c|}{|\chi^{-1}(c)|} \geq \frac{m-k(\ell+d(c))}{m}$
 where $\ell\coloneqq\ell_T \coloneqq 2|V_{\leq 1}(T)| + |V_{\geq 3}(T)|$
 and
 \[d(c)\coloneqq d_T(c) \coloneqq \begin{cases}
                     -2 &\text{if }\deg_T(c)=1,\\
                     0  &\text{if }\deg_T(c)=2,\\
                     -1 &\text{otherwise}
                    \end{cases}.\]
 First observe that this proves the lemma since $k(\ell+d(c)) \leq k\ell < m$, and thus $|W_c| > 0$ for all $c \in C$.
 
 Consider the base case $|V(T)| = 1$ and assume $C= \{c\}$.
 In this case $\ell=2$, $\deg(c)=0$ and $d(c) = -1$.
 Also, for each vertex $v \in V(G)\setminus V_c(H_1,\ldots,H_k)$, the graph $H_{k+1}\coloneqq(\{v\},\emptyset)$ is vertex-disjoint to $H_1,\ldots,H_k$ and agrees with $T$.
 Hence, $\frac{|W_c|}{|\chi^{-1}(c)|} \geq \frac{m-k(2-1)}{m} =
 \frac{m-k(\ell+d(c))}{m}$.

 For the inductive step assume that $|V(T)| \geq 2$ and let $c \in C$. We distinguish two cases depending on the degree of $c$ in the tree $T$.
 \begin{cs}
  \case{$\deg_T(c)=1$}
   Let $c_1,\ldots,c_s=c$ be the unique path in $T$ such that $c_1 \neq c$ is a color with $\deg_T(c_1)\neq 2$ and $\deg_T(c_i) = 2$ for all $i \in \{2,\ldots,s-1\}$ (see also Figure \ref{fig:extension-set-induced-paths}).
   Consider the subtree $T' \coloneqq T - \{c_2,\ldots,c_s\}$ obtained from $T$ by removing all colors in the path excluding $c_1$.
   Define $m' \coloneqq \min_{c\in V(T')}|\chi^{-1}(c)|$, $\ell'\coloneqq \ell_{T'}$ and $d'(c_1) \coloneqq d_{T'}(c_1)$.
   Also define $G'\coloneqq G[\bigcup_{c \in V(T')}\chi^{-1}(c)]$ and $H_i'\coloneqq H_i[V(G')]$ for $i \in [k]$.
   
   If $\deg_T(c_1)=3$ (and thus $\deg_{T'}(c_1)=2$), then $\ell'=\ell-3$ and $d'(c_1)=0=d(c)+2$.
   Otherwise $\deg_T(c_1) \notin \{2,3\}$ (and thus $\deg_{T'}(c_1) \notin \{1,2\}$) and $\ell'=2=\ell-2$ and $d'(c_1)=-1=d(c)+1$.
   In total, $\ell'+d'(c_1)=\ell+d(c)-1$.
   
   We define $W_{c_1}' \coloneqq W_{c_1}(H_1',\dots,H_k')$ with respect to $G'$ and $T'$.
   By the induction hypothesis, $\frac{|W_{c_1}'|}{|\chi^{-1}(c_1)|}\geq\frac{m'-k(\ell'+d'(c_1))}{m'}\geq\frac{m-k(\ell+d(c)-1)}{m}$.
   Now, define $\widetilde T\coloneqq T[\{c_1,\ldots,c_s\}]$, $\widetilde G \coloneqq G[\bigcup_{c\in V(\widetilde T)}\chi^{-1}(c)]$ and $\widetilde H_i\coloneqq H_i[V(\widetilde G)]$ for $i \in [k]$.
   We apply Lemma \ref{la:find-disjoint-paths} to $\widetilde G$, $\widetilde T$, $\widetilde H_i$, $i \in [k]$, and $X \coloneqq W_{c_1}'$.
   Note that $W_{X,c_s}\subseteq W_{c_s}$.
   Thus,
   \begin{align*}
    \frac{|W_{c_s}|}{|\chi^{-1}(c_s)|} &\geq \frac{|W_{X,c_s}|}{|\chi^{-1}(c_s)|} \geq \frac{|W_{c_1}'|}{|\chi^{-1}(c_1)|} - \frac{k}{m}\\
                                       &\geq \frac{m-k(\ell+d(c)-1)}{m}-\frac{k}{m}\\
                                       &=\frac{m-k(\ell+d(c))}{m}
   \end{align*}
   which completes this case.
  \case{$\deg_T(c)\geq2$}
   Let $Z_1,\ldots,Z_s$ be the connected components of $T - \{c\}$.
   Note that $s = \deg_T(c)$.
   Let $T_i \coloneqq T[Z_i \cup \{c\}]$ for all $i \in [s]$.
   Observe that $V(T_i) \cap V(T_j) = \{c\}$ for all distinct $i,j \in [s]$.
   Let $G_i \coloneqq G[\bigcup_{c\in V(T_i)}\chi^{-1}(c)]$ for $i \in [s]$.
   Let $H_{i,j} \coloneqq H_j[V(G_i)]$ for $i \in [s]$ and $j \in [k]$.
   Note that $H_{i,j}$ agrees with $T_i$ for all $i \in [s]$ and $j \in [k]$.
   Let $W_{i,c} \coloneqq W_c(H_{i,1},\dots,H_{i,k})$, $i\in[s]$, be the extension set with respect to $G_i$ and $T_i$.
   Finally, define $m_i \coloneqq \min_{c\in V(T_i)}|\chi^{-1}(c)|$, $\ell_i \coloneqq \ell_{T_i}$ and $d_i(c) \coloneqq d_{T_i}(c)$ be the corresponding parameters for each $T_i$, $i \in [s]$.
   
   If $\deg_T(c) = s = 2$, then $\ell=\sum_{i=1}^s(\ell_i - 2)$ and $d(c) = 0 = \sum_{i=1}^s(d_i(c)+2)$.
   Otherwise, $\deg_T(c) = s \geq 3$ in which case $\ell= 1+\sum_{i=1}^s(\ell_i-2)$ and $d(c)=-1=-1+\sum_{i=1}^s(d_i(c)+2)$.
   In both cases, $\ell+d(c)=\sum_{i=1}^s(\ell_i+d_i(c))$.
   
   By the induction hypothesis, $\frac{|W_{i,c}|}{|\chi^{-1}(c)|} \geq \frac{m_i-k(\ell_i+d_i(c))}{m_i} \geq \frac{m-k(\ell_i+d_i(c))}{m}$ for all $i \in [s]$.
   Moreover, $\bigcap_{i \in [s]} W_{i,c} \subseteq W_c$.
   Together, this means
   \begin{align*}
    \frac{|W_c|}{|\chi^{-1}(c)|} &\geq \frac{\left|\bigcap_{i \in [s]} W_{i,c}\right|}{|\chi^{-1}(c)|} \geq 1 - \sum_{i \in [s]} \left(1 - \frac{|W_{i,c}|}{|\chi^{-1}(c)|}\right)\\
                                 &\geq 1 - \sum_{i \in [s]} \left(1 - \frac{m-k(\ell_i+d_i(c))}{m}\right)\\
                                 &= \frac{m-k(\ell_1+\ldots+\ell_s+d_1(c)+\ldots+d_s(c))}{m}\\
                                 &=\frac{m-k(\ell+d(c))}{m}.
   \end{align*}
 \end{cs}
\end{proof}

\begin{proof}[Proof of Lemma \ref{la:small-color-classes-excluded-kh}]
 Consider the graph $H \coloneqq G[[\chi,V_2]]$ which is connected.
 Let $v_1,\dots,v_h \in V_1$ be distinct vertices and let $w_1,\dots,w_h \in V_2$ such that $v_iw_i \in E(G)$.
 Note that $[w_i]_\chi \subseteq N(v_i)$ for all $i \in [h]$ since $\chi$ is stable with respect to the Color Refinement algorithm.
 Also define $c_i = \chi(w_i)$.
 Now let $T\subseteq H$ be a Steiner tree for $\{c_1,\ldots,c_h\}$,
 i.e., a tree that contains
 all the vertices $c_1,\dots,c_h$ and is minimal with respect to the subgraph relation.
 Hence, $T$ is a tree with $c_1,\dots,c_h \in V(T)$ and $|V_{\leq 1}(T)| \leq h$.
 This also implies that $|V_{\geq 3}(T)| \leq h$.
 
 Now let $\ell \coloneqq 2|V_{\leq 1}(T)|+|V_{\geq 3}(T)| \leq 3h$.
 We have that $m\coloneqq \min_{c \in V(T)} |\chi^{-1}(c)|\geq 3h^3$.
 Also note that $G[\chi^{-1}(t_1),\chi^{-1}(t_2)]$ is biregular and non-trivial for all $t_1t_2 \in E(T)$.
 By Lemma \ref{la:find-disjoint-trees}, there are $r\coloneqq\lfloor\frac{m}{\ell}\rfloor\geq h^2$ pairwise vertex-disjoint trees $H_1,\ldots,H_r$ that agree with $T$.
 But this gives a topological subgraph $K_h$ of the graph $G$.
 For each unordered pair $v_iv_j$, $i,j \in [h]$, and each $H_p$, $p \in [r]$, there is a path in the graph $H_p$ from a vertex $w_i'\in[w_i]_\chi\subseteq N(v_i)$ to a vertex $w_j'\in[w_j]_\chi\subseteq N(v_j)$.
 Therefore, for each unordered pair $v_iv_j$, $i,j \in [h]$, there is a path from $v_i$ to $v_j$ in $G$ and these paths are internally vertex disjoint (since $H_1,\ldots,H_r$ are pairwise vertex-disjoint trees).
\end{proof}

\subsection{Finding an Initial Color Class}
\label{sec:initial-color}

Next, we give a proof for Theorem \ref{thm:initial-color}.
The proof builds on the $2$-dimensional Weisfeiler-Leman algorithm.
Towards this end, we first introduce some additional notation.

Let $G$ be a graph and let $\chi \coloneqq \WL{2}{G}$ the coloring computed by the $2$-dimensional Weisfeiler-Leman algorithm.
We refer to $C_V \coloneqq C_V(G,\chi)\coloneqq\{\chi(v,v) \mid v \in V(G)\}$ as the set of \emph{vertex colors} and
$C_E\coloneqq C_E(G,\chi) \coloneqq \{\chi(v,w) \mid vw \in E(G)\}$ as the set of \emph{edge colors}.
For a vertex color $c\in C_V(G,\chi)$, we define $V_c\coloneqq V_c(G,\chi)\coloneqq\{v\in V(G)\mid \chi(v,v)=c\}$
as the set of all vertices with color $c$.
Similar, for an edge color $c\in C_E(G,\chi)$, we define $E_c\coloneqq E_c(G,\chi)\coloneqq\{v_1v_2\in E(G)\mid\chi(v_1,v_2)=c\}$.
Let $C \subseteq C_E$ be a set of edge colors.
We define the graph $G[C]$ with vertex set
\begin{equation}
 \label{eq:color-graph}
 V(G[C]) \coloneqq \bigcup_{c \in C}\bigcup_{e \in E_c} e \quad\quad\text{and}\quad\quad E(G[C]) \coloneqq \bigcup_{c \in C}E_c.
\end{equation}

Let $A_1,\dots,A_\ell$ be the vertex sets of the connected components of $G[C]$.
We also define the graph $G/C$ as the graph obtained from contracting every set $A_i$ to a single vertex.
Formally, we set
\begin{equation}
 \label{eq:color-factor}
 \begin{aligned}
  V(G/C) &\coloneqq \{\{v\} \mid v \in V(G) \setminus V(G[C])\} \cup \{A_1,\dots,A_\ell\} \quad\text{and}\\
  E(G/C) &\coloneqq \{X_1X_2 \mid \exists v_1 \in X_1,v_2 \in X_2\colon v_1v_2 \in E(G)\}.
 \end{aligned}
\end{equation}
Observe that $G/C$ is a minor of $G$ for every set of edge colors $C \subseteq C_E$.
We usually only consider the case $C = \{c\}$, and we write $G[c]$ and $G/c$ instead of $G[\{c\}]$ and $G/\{c\}$.

\begin{lemma}[{\cite[Section I]{Weisfeiler76}}]
 \label{la:factor-graph-2-wl}
 Let $G$ be a graph and $C \subseteq C_E$ be a set of edge colors.
 Define
 \[(\chi/C)(X_1,X_2) \coloneqq \{\!\!\{\chi(v_1,v_2) \mid v_1 \in X_1, v_2 \in X_2\}\!\!\}\]
 for all $X_1,X_2 \in V(G/C)$.
 Then $\chi/C$ is a stable coloring of the graph $G/C$ with respect to the $2$-dimensional Weisfeiler-Leman algorithm.
 
 Moreover, for all $X_1,X_2,X_1',X_2' \in V(G/C)$, either it holds $(\chi/C)(X_1,X_2) = (\chi/C)(X_1',X_2')$ or $(\chi/C)(X_1,X_2) \cap (\chi/C)(X_1',X_2') = \emptyset$.
\end{lemma}

For a more recent reference we also point the reader to \cite[Theorem 3.1.11]{ChenP19}.
As before, if $C = \{c\}$, we write $\chi/c$ instead of $\chi/\{c\}$.

For every edge color $c$, the endvertices of all $c$-colored edges have the same vertex colors, that is,
for all edges $vw,v'w'\in E(G)$ with $\chi(v,w) = \chi(v',w') = c$ we have $\chi(v,v) = \chi(v',v')$ and $\chi(w,w) = \chi(w',w')$.
This implies $1 \leq |C_V(G[c],\chi)| \leq 2$.
We say that $G[c]$ is \emph{unicolored} if $|C_V(G[c],\chi)| = 1$.
Otherwise $G[c]$ is called \emph{bicolored}.

The basic strategy for the proof of Theorem \ref{thm:initial-color} is to color the input graph with the coloring $\chi$ computed by the $2$-dimensional Weisfeiler-Leman algorithm.
The goal is to find a color class $X = V_d$ (for some vertex color $d$) such that $X \subseteq \cl_t^{(G,\chi)}(v)$ for every $v \in X$.
We prove the existence of such a color class by a complicated case distinction depending on which types of graphs $G[c]$ occur within the graph $G$.
For example, a simple case is that there is a unicolored graph $G[c]$ that is connected.
Note that $G[c]$ is $r$-regular for some number $r \geq 1$, and $r \leq a h \sqrt{\log h}$ by Theorem \ref{thm:average-degree-excluded-minor}.
So we can choose $X = V(G[c])$ since the closure always contains bounded-degree components (of some edge color) assuming at least one vertex of the component is individualized.

For the other cases, the proof turns out to be significantly more complicated.
Here, we also need to rely on recursive approaches.
A first, simple idea is that in certain cases, we are able to distinguish between some vertices that receive the same color from the $2$-dimensional Weisfeiler-Leman algorithm.
In this case, we update the coloring $\chi$ accordingly (in an isomorphism-invariant way) and restart the entire algorithm using the updated, finer coloring.

A second idea that is used in our algorithm is to consider an edge color $c$, recursively obtain a set $X'$ for the graph $G/c$, and then construct the set $X$ from $X'$.
This approach is actually the reason why Theorem \ref{thm:initial-color} provides a pair-colored graph $(G',\chi')$: we need to ensure that $G/c$ is a subgraph of $G'$ so that properties for the closure can be lifted from $G/c$ to $G'$.

Let us now dive into the technical details of the proof.
The next two lemmas investigate properties of connected components of bicolored graphs $G[c]$ for an edge color $c$.
In particular, we show that if $G[c]$ is a bicolored connected graph with vertex color classes $V_{d_1}$ and $V_{d_2}$, then there is $i \in \{1,2\}$ such that $V_{d_i} \subseteq \cl_t^{(G,\chi)}(v)$ for every $v \in V_{d_i}$.
Again, recall the definition of the constant $a$ from Theorem \ref{thm:average-degree-excluded-minor}.

\begin{lemma}
 \label{la:hat-graph-excluded-minor}
 Let $G = (V_1,V_2,E)$ be a connected, bipartite graph that excludes $K_h$ as a minor and define $\chi \coloneqq \WL{2}{G}$.
 Suppose that $\chi(v_1,v_2) = \chi(v_1',v_2')$ for all $(v_1,v_2),(v_1',v_2') \in V_1 \times V_2$ with $v_1v_2,v_1'v_2'\in E$.
 Also assume that $|V_2| > (a h \sqrt{\log h})\cdot|V_1|$.
 Let 
 \[E^{*} \coloneqq \left\{v_1v_1' \in \binom{V_1}{2} \mathrel{}\middle|\mathrel{} \exists v_2 \in V_2 \colon v_1v_2,v_1'v_2 \in E(G)\right\}.\]
 Then there are colors $c_1,\dots,c_r \in \chi(V_1^{2})$ such that
 \begin{enumerate}
  \item $E^{*} = \bigcup_{i \in [r]} E_{c_i}$ where $E_{c_i}\coloneqq\{w_1w_2\in V(G)^2\mid\chi(w_1,w_2)=c_i\}$,
  \item $H \coloneqq (V_1,E^{*})$ is connected, and
  \item $H_i$ is a minor of $G$ for all $i \in [r]$ where $H_i = (V_1,E_{c_i})$.
 \end{enumerate}
\end{lemma}

\begin{proof}
 Clearly, $H$ is connected because $G$ is connected.
 Since $\chi(v_1,v_2) = \chi(v_1',v_2')$ for all $(v_1,v_2),(v_1',v_2') \in V_1 \times V_2$ with $v_1v_2,v_1'v_2'\in E$, it follows that $\chi(v_1,v_1) = \chi(v_1',v_1')$ for all $v_1,v_1' \in V_1$, and $\chi(v_2,v_2) = \chi(v_2',v_2')$ for all $v_2,v_2' \in V_2$.
 So $G$ is biregular which implies that $\deg(v_1)\cdot|V_1| = |E| = \deg(v_2)\cdot|V_2|$ for all $v_1 \in V_1$ and $v_2 \in V_2$.
 Since $|V_2| > |V_1|$ we conclude that $\deg(v_1) \neq \deg(v_2)$, and hence $\chi(v_1,v_1) \neq \chi(v_2,v_2)$ for all $v_1 \in V_1$ and $v_2 \in V_2$.
 All together, this means that there is some vertex color $d \in C_V$ such that $V_d(G,\chi) = V_1$.
 Since the $2$-dimensional Weisfeiler-Leman algorithm distinguishes pairs of distinct vertices with a common neighbor from other pairs of vertices, we conclude that there are colors $c_1,\dots,c_r \in \chi(V_1^{2})$ such that $E^{*} = \bigcup_{i \in [r]} E_{c_i}$.
 
 So fix some $i \in [r]$ and consider the bipartite graph $B = (V_2,E_{c_i},E(B))$ where $E(B) \coloneqq \{(v_2,v_1v_1') \mid v_2 \in N_G(v_1) \cap N_G(v_1')\}$.
 By the properties of the $2$-dimensional Weisfeiler-Leman algorithm the graph $B$ is biregular.
 So it follows from Hall's Marriage Theorem that $B$ contains a matching $M$ of size $\min(|V_2|,|E_{c_i}|)$ as explained in the preliminaries.
 
 If $|V_2| \geq |E_{c_i}|$ then each pair $v_1v_1' \in E_{c_i}$ is matched to a vertex $v_2 \in V_2$ (i.e., $(v_2,v_1v_1') \in M$) such that $v_1v_2,v_1'v_2 \in E(G)$.
 It follows that $H_i$ is a minor of $G$.
 
 Otherwise $|V_2| < |E_{c_i}|$.
 Let $F_i \subseteq E_{c_i}$ be those vertices that are matched by the matching $M$ in the graph $B$.
 Then $H_i' = (V_1,F_i)$ is a minor of $G$, and thus it excludes $K_h$ as a minor.
 However,
 \[\frac{1}{|V(H_i')|} \sum_{v_1 \in V(H_i')} \deg_{H_i'}(v_1) = \frac{2|F_i|}{|V_1|} = \frac{2|V_2|}{|V_1|} > 2a h \sqrt{\log h}\]
 which contradicts Theorem \ref{thm:average-degree-excluded-minor}.
\end{proof}

\begin{lemma}
 \label{la:closure-bipartite}
 Let $t \geq (a h \sqrt{\log h})^{2}$.
 Let $G = (V_1,V_2,E)$ be a connected bipartite graph that excludes $K_h$ as a minor and define $\chi \coloneqq \WL{2}{G}$.
 Suppose that $\chi(v_1,v_2) = \chi(v_1',v_2')$ for all $(v_1,v_2),(v_1',v_2') \in V_1 \times V_2$ with $v_1v_2,v_1'v_2'\in E$.
 Also assume that $|V_1| \leq |V_2|$.
 Then $V_1 \subseteq \cl_t^{(G,\chi)}(v)$ for all $v \in V_1 \cup V_2$.
\end{lemma}

\begin{proof}
 The graph $G$ is biregular and it holds that
 $\deg(v_1)\cdot|V_1| = \deg(v_2)\cdot|V_2|$ for all $v_1 \in V_1$
 and $v_2 \in V_2$.  Hence, $\deg(v_2) \leq a h \sqrt{\log h}$ for all
 $v_2 \in V_2$ by Theorem \ref{thm:average-degree-excluded-minor}.
 This means $\cl_t^{(G,\chi)}(v) \cap V_1 \neq \emptyset$, because
 either $v\in\cl_t^{(G,\chi)}(v)\cap V_1$ or $v\in V_2$ and
 $N_G(v)\subseteq \cl_t^{(G,\chi)}(v) \cap V_1$.

 First suppose that $|V_2| \leq (a h \sqrt{\log h})\cdot |V_1|$.
 Then $\deg(v_1) = \deg(v_2)\frac{|V_2|}{|V_1|} \leq t$ and $\deg(v_2)\leq t$ for all $v_1 \in V_1,v_2\in V_2$.
 It follows that $\cl_t^{(G,\chi)}(v)=V(G)$.
 
 So assume that $|V_2| > a h \sqrt{\log h} |V_1|$.
 By Lemma \ref{la:hat-graph-excluded-minor}, there are colors $c_1,\dots,c_r \in \chi(V_1^{2})$ such that
 \begin{enumerate}
  \item $H_i$ excludes $K_h$ as a minor for all $i \in [r]$ where $H_i = (V_1,E_{c_i})$ and $E_{c_i}\coloneqq\{v_1v_2\in V(G)^2\mid\chi(v_1,v_2)=c_i\}$, and
  \item $H = (V_1,\bigcup_{i \in [r]} E_{c_i})$ is connected.
 \end{enumerate}
 Note for all $i\in[r]$, the graph $H_i$ is $d$-regular for some $d$,
 and by Theorem \ref{thm:average-degree-excluded-minor} we have $d \leq a h \sqrt{\log h} \leq t$. This implies that $V_1 \subseteq
 \cl_t^{(G,\chi)}(v_1)$ for all $v_1\in V_1$, and since $V_1\cap
 \cl_t^{(G,\chi)}(v)\neq\emptyset$  for all $v\in V_1\cup V_2$, it follows that $V_1 \subseteq
 \cl_t^{(G,\chi)}(v)$.
\end{proof}

Let $A$ be the vertex set of a connected component of $G[c]$.
We define a size parameter for the graph $G[c]$ as
\[s(c) \coloneqq \min_{d \in C_V(G[c],\chi)} |A \cap V_d|.\]
Note that this is well-defined since every two connected components of $G[c]$ are equivalent with respect to the $2$-dimensional Weisfeiler-Leman algorithm (since the $2$-dimensional Weisfeiler-Leman algorithm ``detects'' components of graphs).

\begin{proof}[Proof of Theorem \ref{thm:initial-color}]
 First observe that $t \geq \max\{(ah\sqrt{\log h})^2,ah^3,3h^3\}$ (recall that we assumed $a \geq 2$).
 Let $\chi \coloneqq \WL{2}{G}$ be the coloring computed by the $2$-di\-men\-sion\-al Weis\-fei\-ler-Leman algorithm for the graph $G$.
 The algorithm works recursively and essentially distinguishes between two cases (but for both cases there are several subcases).
 Also, in some situations, the algorithm may find an isomorphism-invariant vertex-coloring $\chi_V$ that is strictly finer than the one induced by $\chi$.
 In this case, the algorithm is always restarted on the current graph and the coloring $\chi$ is updated accordingly (by running the $2$-dimensional Weisfeiler-Leman algorithm on $G$ with vertex-coloring $\chi_V$).
 Note that such a restart can only occur at most $n$ times where $n$ denotes the number of vertices of $G$.
 
 \medskip
 The algorithm distinguishes between two cases depending on whether there is an edge color $c^{E} \in C_E\coloneqq C_E(G,\chi)$ such that $s(c^{E}) \leq ah^3$ where $a$ is the constant from Theorem \ref{thm:average-degree-excluded-minor}.
 
 In the first case, we assume that such a color exists.
 We choose an edge color $c^{E} \in C_E$ such that $s(c^{E}) \leq ah^3$ and define $F \coloneqq G/c^{E}$
 (recall Equation \eqref{eq:color-factor} for the definition; to ensure that the color in $C_E$ is chosen in an isomorphism-invariant way, the algorithm chooses the smallest color in $C_E \subseteq \NN$ according to the ordering of natural numbers).
 The algorithm recursively computes an isomorphism-invariant graph $(F',\chi_F')$ and a vertex color $c^F \in C_V(F',\chi_F')$ that satisfies Properties \ref{item:initial-color-1} and \ref{item:initial-color-2} for the input graph $F$.
 (If $F$ contains a minor $K_h$, then $G$ also contains a minor $K_h$ since $F$ is a minor of $G$.)
 If $V_{c^F} \subseteq V(G)$, then the algorithm simply returns $(F',\chi_F')$ and the color $c^F$.
 
 Otherwise, $V_{c^F} \cap \{A_1,\dots,A_\ell\} \neq \emptyset$ where $A_1,\dots,A_\ell$ are the vertex sets of the connected components of $G[c^{E}]$.
 Let $d\coloneqq\argmin_{d\in C_V(G,\chi),A_i\cap V_d\neq\emptyset}|A_i\cap V_d|$ for some $i \in [\ell]$.
 This means $s(c^{E}) = |A_i \cap V_d|$.
 The algorithm constructs $G'$ where
 \[V(G') \coloneqq V(F') \uplus \bigcup_{A \in V_{c^F} \cap \{A_1,\dots,A_\ell\}} (A \cap V_d)\]
 and
 \[E(G') \coloneqq E(F') \cup \{Av \mid A \in V_{c^F} \cap \{A_1,\dots,A_\ell\}, v \in A \cap V_d\}.\]
 Also, $\chi'(v,w) \coloneqq (\chi_F'(v,w),0)$ for every $v,w \in V(F')$, $\chi'(w,v) = \chi'(v,w) \coloneqq (1,1)$ for all distinct $v \in V(G')$, $w \in V(G') \setminus V(F')$, and $\chi'(v,v) \coloneqq c \coloneqq (0,1)$ for every $v \in V(G') \setminus V(F')$.
 Clearly, $(G',\chi')$ is constructed in an isomorphism-invariant manner.
 The algorithm returns $(G',\chi')$ together with the vertex color $c$.
 By definition, $V_{c} \subseteq V_d \subseteq V(G)$.
 Moreover, $V_{c} \subseteq \cl_t^{(G',\chi')}(v)$ for all $v \in V_{c}$ 
 since  $|N_{G'}(A_i) \cap V_{c}| \leq |A_i \cap V_d| = s(c^{E}) \leq ah^3$ for all $A_i \in V_{c^{F}}$.
 This completes the first case.
 
 \medskip
 
 In the second case, $s(c^{E}) > ah^3$ for every edge color $c^{E} \in C_E$.
 If $|V(G)| = 1$, the problem is trivial.
 Otherwise, let $c \coloneqq \argmin_{c \in C_V}|V_c|$ be the color of the smallest color class (if this color is not unique, then the algorithm chooses the smallest color in $C_V \subseteq \NN$ with minimal color class size).
 
 Now let $c^{E} \in C_E$ be an edge color defined in such a way that
 either $V_c = V(G[c^{E}])$, or $V_c \subseteq V(G[c^{E}])$ and there
 is no edge color $c_2^{E}$ such that $V_c = V(G[c_2^{E}])$.
 
 First suppose that $G[c^{E}]$ is connected.
 If $G[c^{E}]$ is unicolored, then $V_c \subseteq \cl_t^{(G,\chi)}(v)$ for every $v \in V_c$ by Theorem \ref{thm:average-degree-excluded-minor}.
 So suppose that $G[c^{E}]$ is bicolored.
 Let $c' \in C_V$ be the second vertex color that appears in $G[c^{E}]$, i.e., $V_c\cup V_{c'}=V(G[c^E])$.
 Note that $G[c^{E}]$ is a bipartite graph with bipartition $(V_c,V_{c'})$ and $|V_c| \leq |V_{c'}|$.
 So $V_c \subseteq \cl_t^{(G,\chi)}(v)$ for every $v \in V_c$ by Lemma \ref{la:closure-bipartite}.
 In this case, we are done and return $(G,\chi)$ together with the vertex color $c$.
  
 So assume that $G[c^{E}]$ is not connected and let $A_1,\dots,A_\ell$ be the vertex sets of the connected components of $G[c^{E}]$.
 Note that $|A_i \cap V_c| = s(c^{E}) > ah^3$.
 Now let $i \in [\ell]$ and $v \in A_i$.
 Then $A_i \cap V_c \subseteq \cl_t^{(G,\chi)}(v)$ by Lemma \ref{la:closure-bipartite}.
 In particular, $\cl_t^{(G,\chi)}(v_1) = \cl_t^{(G,\chi)}(v_2)$ for all $v_1,v_2 \in A_i \cap V_c$.
 Let $D_i \coloneqq \cl_t^{(G,\chi)}(v)$ for some $v \in A_i \cap V_c$, $i \in [\ell]$.
 
 Without loss of generality assume that $|D_i \cap V_{c'}| = |D_j \cap V_{c'}|$ for all $i,j \in [\ell]$ and $c' \in C_V$.
 (Otherwise define $\chi_V(v) \coloneqq (\chi(v,v),0)$ for all $v \in V(G) \setminus V_c$ and $\chi_V(v) \coloneqq (\chi(v,v),(|D_i \cap V_{c'}|)_{c' \in C_V})$ for all $v \in A_i \cap V_c$ and $i \in [\ell]$.
 Then $\chi_V$ is isomorphism-invariant and strictly refines the vertex-coloring induced by $\chi$, and the algorithm is restarted as discussed above.)
 Note that for every $i,i' \in [\ell]$ it holds that
 \begin{equation}
  \label{eq:cover-component}
  A_{i'} \cap D_i \neq \emptyset \;\;\;\implies\;\;\; A_{i'} \cap V_c \subseteq D_i
 \end{equation}
 by Lemma \ref{la:closure-bipartite}.
 
 Let $R \coloneqq \{(i,i') \in [\ell] \mid A_{i'} \cap V_c \subseteq D_i\}$.
 Clearly, $R$ is reflexive and transitive.
 
 \begin{claim}
  If $R$ is not symmetric, one can compute in polynomial time an
  isomorphism-invariant vertex-coloring $\chi_V$ that is strictly finer than the one induced by $\chi$.
 \end{claim}
 \begin{claimproof}
  Since $R$ is not symmetric, there are distinct $i,i' \in [\ell]$ such that $(i,i') \in R$ and $(i',i) \notin R$.
  Now, consider the directed graph $([\ell],R)$ (ignoring self-loops).
  Since $R$ is transitive it follows that $([\ell],R)$ is not strongly connected.
  Let $M \subseteq [\ell]$ denote the set of those vertices that appear in a maximal strongly connected component of $([\ell],R)$ (i.e., a strongly connected component without outgoing edges).
  Also, let $M_c \coloneqq \bigcup_{i \in M} A_i \cap V_c$.
  Clearly, $M_c$ is defined in an isomorphism-invariant manner and $\emptyset \neq M_c \subsetneq V_c$.
  We define $\chi_V(v) \coloneqq (\chi(v,v),0)$ for all $v \in V(G) \setminus M_c$ and $\chi_V(v) \coloneqq (\chi(v,v),1)$ for all $v \in M_c$.
  It is easy to see that all objects can be computed in polynomial time.
 \end{claimproof}
 
 With the last claim in mind, we may assume that $R$ is symmetric (otherwise the algorithm is restarted).
 So overall, $R$ is an equivalence relation.
 Let $\CA_1,\dots,\CA_r$ be the equivalence classes of $R$.
 If $r = 1$, then $V_c \subseteq D_i$ for all $i \in [\ell]$ and the algorithm returns $(G,\chi)$ together with the vertex color $c$.
 
 \begin{figure}
  \centering
  \begin{tikzpicture}
   \draw[violet,thick,dashed,rounded corners,fill=violet,fill opacity=0.3] (0.4, -0.2) rectangle (9.2, 0.2);
   \node at (10,0) {{\color{violet} $V_c$}};
   
   \foreach \i in {1,2,3,4,5,6}{
    \node[tinyvertex] (v1\i) at (1.6*\i - 1.0, 0) {};
    \node[tinyvertex] (v2\i) at (1.6*\i - 0.6, 0) {};
    
    \node[tinyvertex] (v3\i) at (1.6*\i - 1.2, 2.4) {};
    \node[tinyvertex] (v4\i) at (1.6*\i - 0.8, 2.4) {};
    \node[tinyvertex] (v5\i) at (1.6*\i - 0.4, 2.4) {};
    
    \draw[thick] (v1\i) edge (v3\i);
    \draw[thick] (v1\i) edge (v4\i);
    \draw[thick] (v1\i) edge (v5\i);
    \draw[thick] (v2\i) edge (v3\i);
    \draw[thick] (v2\i) edge (v4\i);
    \draw[thick] (v2\i) edge (v5\i);
    
    \draw[blue,thick,dashed,rounded corners] (1.6*\i - 0.3, -0.4) rectangle (1.6*\i - 1.3, 2.6);
    \node at (1.6*\i - 0.8, -0.8) {{\color{blue} $A_\i$}};
   }
   
   \draw[red,thick,dashed,rounded corners] (0.1, -1.2) rectangle (3.1, 1.2);
   \node at (1.6,-1.6) {{\color{red} $D_1 = D_2$}};
   
   \draw[red,thick,dashed,rounded corners] (3.3, -1.2) rectangle (6.3, 1.2);
   \node at (4.8,-1.6) {{\color{red} $D_3 = D_4$}};
   
   \draw[red,thick,dashed,rounded corners] (6.5, -1.2) rectangle (9.5, 1.2);
   \node at (8.0,-1.6) {{\color{red} $D_5 = D_6$}};
   
   \node at (-0.4,1.2) {$G[c^{E}]$};
   
   \node at (-0.4,3.2) {$G - V_c$};   
   \draw[dotted,thick] (4.8,1.6) arc (270:287:20);
   \draw[dotted,thick] (4.8,1.6) arc (270:253:20);
  \end{tikzpicture}
  \caption{Visualization of the proof in the case $s(c^{E}) \geq ah^{3}$ (this condition is omitted for visualization purposes).
  The figure shows the part $G[c^{E}]$ with six connected components $A_1,\dots,A_6$.
  Moreover, $D_1 = D_2$, $D_3 = D_4$ and $D_5 = D_6$, and $\CA_1 = \{1,2\}$, $\CA_2 = \{3,4\}$ and $\CA_3 = \{5,6\}$.}
  \label{fig:initial-color}
 \end{figure}
 
 Hence, assume that $r \geq 2$.
 This is visualized in Figure \ref{fig:initial-color}.
 Note that $D_i = D_{i'}$ for all $(i,i') \in R$.
 The rest of the proof is devoted to computing a vertex-coloring that is strictly finer than the one induced by $\chi$ (which results again in a restart).
 
 A partition $\CP = \{P_1,\dots,P_q\}$ of the set $[\ell]$ \emph{refines} the partition $\{\CA_1,\dots,\CA_r\}$, denoted $\CP \preceq \{\CA_1,\dots,\CA_r\}$, if for every $j \in [q]$ there is some $i \in [r]$ such that $P_j \subseteq \CA_i$.
 
 \begin{claim}
  \label{cl:factor-for-separators}
  There is a partition $\CP = \{P_1,\dots,P_q\} \preceq \{\CA_1,\dots,\CA_r\}$
  and a graph $G_\CP = (\CP,E_\CP)$ such that $G_\CP$
  is a minor of $G$ and there are distinct $i,i' \in [r]$ and
  $P,P' \in \CP$ such that $P \subseteq \CA_i$,
  $P' \subseteq \CA_{i'}$ and $PP' \in E_{\CP}$.
  Moreover, the partition $\CP$ and the graph $G_\CP$ are isomorphism-invariant and can be computed in polynomial time.
 \end{claim}
 \begin{claimproof}
  We first construct an inclusionwise maximal set of edge colors $C^* \subseteq C_E(G,\chi)$ such that
  \begin{enumerate}[label = (\alph*)]
   \item $c^E \in C^*$, and
   \item\label{item:factor-for-separators-2} for every connected component $B$ of $G[C^*]$ there is some $i \in [\ell]$ such that $A_i \subseteq B$ and $B \cap V_c \subseteq D_i$.
  \end{enumerate}
  Observe that such a set $C^*$ can easily be constructed by a greedy algorithm that initially sets $C^* \coloneqq \{c^E\}$ and keeps adding colors as long as Condition \ref{item:factor-for-separators-2} is satisfied.
  To ensure isomorphism-invariance, the greedy algorithm always adds the smallest color in $C_E(G,\chi) \setminus C^* \subseteq\NN$ that does not violate Condition \ref{item:factor-for-separators-2}.
  
  Now let $B_1,\dots,B_q$ denote the connected components of $G[C^*]$.
  We define $P_j \coloneqq \{i \in [\ell] \mid A_i \subseteq B_j\}$.
  Let $\CP \coloneqq \{P_1,\dots,P_q\}$.
  Clearly, $\CP$ is a partition of $[\ell]$.
  Also $\CP \preceq \{\CA_1,\dots,\CA_r\}$ by Condition \ref{item:factor-for-separators-2}.
  
  To define the graph $G_\CP$ we associate the elements of $\CP$ with the sets $B_1,\dots,B_q$.
  Consider the graph $F \coloneqq G/C^*$ and let $\chi_F \coloneqq \chi/C^*$ as defined in Lemma \ref{la:factor-graph-2-wl}.
  Note that $\chi_F$ is stable with respect to the $2$-dimensional Weisfeiler-Leman algorithm for the graph $F$ by Lemma \ref{la:factor-graph-2-wl}.
  
  Let $c^F \in C_E(F,\chi_F)$ be an edge color such that $\{B_1,\dots,B_q\} \subseteq V(F[c^F])$.
  Note that such an edge color exists since $F$ is connected and $q \geq r > 1$.
  
  \begin{cs}
   \case{$F[c^F]$ is unicolored}
    We set
    \[E_\CP \coloneqq \left\{P_jP_{j'} \in\binom{\CP}{2} \mathrel{}\middle|\mathrel{} \chi_F(B_j,B_j') = c^F\right\}.\]
    Clearly, $G_\CP$ is isomorphic to $F[c^F]$ and hence, it is a minor of $G$.
    Also, there are distinct $i,i' \in [r]$ and $P,P' \in \CP$ such that $P \subseteq \CA_i$, $P' \subseteq \CA_{i'}$ and $PP' \in E_{\CP}$ by the maximality of the set $C^*$.
   \case{$F[c^F]$ is bicolored}
    Let $U \coloneqq V(F[c^F]) \setminus \{B_1,\dots,B_q\}$.
    Note that $|U| \geq |V_c| \geq \ell a h^{3} \geq qah^3$, because $V_c$ is a color class of $G$ of minimum size.
    We define
    \[E_\CP^+ \coloneqq \left\{P_jP_{j'}\in\binom{\CP}{2} \mathrel{}\middle|\mathrel{} \exists u \in U\colon \chi_F(B_j,u)=\chi_F(B_{j'},u)=c_F\right\}.\]
    Again, there are distinct $i,i' \in [r]$ and $P_j,P_{j'} \in \CP$ such that $P_j \subseteq \CA_i$, $P_{j'} \subseteq \CA_{i'}$ and $P_jP_{j'} \in E_{\CP}^+$ by the maximality of the set $C^*$.
    
    Finally, to obtain a minor of $G$, we further thin out the set $E_\CP^+$ by keeping only those pairs that have the same color as $(P_j,P_{j'})$.
    Formally, we define
    \[E_\CP \coloneqq \left\{P_{j''}P_{j'''}\in\binom{\CP}{2} \mathrel{}\middle|\mathrel{} \chi_F(B_{j''},B_{j'''}) = \chi_F(B_{j},B_{j'})\right\}.\]
    Then $G_\CP$ is a minor of $G$ by Lemma \ref{la:hat-graph-excluded-minor} which is applicable since $|U| \geq qah^3$.
  \end{cs}
 \end{claimproof}
 
 Now let $\CP = \{P_1,\dots,P_q\}$ and $G_{\CP} = (\CP,E_\CP)$ be the objects computed in Claim \ref{cl:factor-for-separators}.
 If $G_{\CP}$ is not regular, then the algorithm computes an isomorphism-invariant refinement of the coloring $\chi$.
 (Actually, this case does not occur by the properties of the $2$-dimensional Weisfeiler-Leman algorithm).
 Hence, $\deg_{G_\CP}(P_j) \leq a h \sqrt{\log h}$ for all $j \in [q]$ by Theorem \ref{thm:average-degree-excluded-minor}.
 
 Let $Z_1^{i},\dots,Z_{k_i}^{i}$ be the connected components of $G - D_i$.
 If there are $i,i' \in [\ell]$ such that $k_i \neq k_{i'}$, then the algorithm computes the vertex-coloring $\chi_V$ defined by $\chi_V(v) \coloneqq (\chi(v,v),k_i)$ for all $v \in A_i \cap V_c$ and $i \in [\ell$] and $\chi_V(v) \coloneqq (\chi(v,v),0)$ for all $v \in V(G) \setminus V_c$.
 Then $\chi_V$ strictly refines the vertex-coloring induced by $\chi$ and the algorithm is restarted.
 
 So suppose $k \coloneqq k_i = k_{i'}$ for all $i,i' \in [\ell]$.
 Then $|N_G(Z_j^{i})| < h$ by Theorem \ref{thm:small-separator-for-t-cr-bounded-closure} for all $j \in [k]$.
 Moreover, it holds that
 \begin{equation}
  \label{eq:component-contains-ai}
  A_{i'} \cap V_c \cap Z_j^{i} \neq \emptyset \;\;\;\iff\;\;\; A_{i'} \subseteq Z_j^{i}
 \end{equation}
 for all $i' \in [\ell]$ and $j \in [k]$ by Equation \eqref{eq:cover-component}.
 Let
 \[\widetilde{E}_{\CP} \coloneqq \{(D_i,A_{i'}) \mid i' \in P_{j'}, i \in P_j, P_jP_{j'} \in E_\CP\}.\]
 For each $i' \in [\ell]$ there are at most $(a h \sqrt{\log h})$ many distinct sets $D \in \{D_1,\dots,D_\ell\}$ such that $(D,A_{i'}) \in \widetilde{E}_{\CP}$.
 This implies $|\widetilde{E}_{\CP}| \leq \ell a h \sqrt{\log h}$.
 We define
 \[Q \coloneqq \left\{(D,j) \in \{D_1,\dots,D_\ell\} \times [k] \mathrel{}\middle|\mathrel{} \exists i' \in [\ell] \colon (D,A_{i'}) \in \widetilde{E}_{\CP} \wedge A_{i'} \subseteq Z_j^i\right\}.\]
 Note that $Q \neq \emptyset$ by Claim \ref{cl:factor-for-separators}
 and that $|Q| \leq |\widetilde{E}_{\CP}| \leq \ell a h \sqrt{\log h}$ by Equation \eqref{eq:component-contains-ai}.
 Now pick $d \in C_V\subseteq\NN$ to be the smallest (according to the ordering of the natural numbers) vertex color such that 
 \[X \coloneqq \bigcup_{(D_i,j) \in Q} N_G(Z_j^{i}) \cap V_{d}(G,\chi)\]
 is not the empty set.
 Then
 \[0 < |X| < \ell a h^2 \sqrt{\log h} \leq |V_c|.\]
 Since $X$ is defined in an isomorphism-invariant manner and $V_c$ forms the smallest vertex color class, this allows us to refine the coloring $\chi$ by taking membership in $X$ into account.
 
 \medskip
 
 This completes the description of the algorithm.
 Clearly, the running time is polynomially bounded in the input size.
 Also, the correctness follows from the arguments provided throughout the description of the algorithm.
\end{proof}

\section{Isomorphism Test for Graphs Excluding a Minor}
\label{sec:isomorphism}

Having presented the key technical tools in the previous section, we are now ready to describe our isomorphism test for graph classes that exclude $K_h$ as a minor.

\begin{theorem}\label{theo:iso}
 Let $h\in\NN$.
 There is an algorithm that, given two connected vertex-colored graphs $G_1,G_2$ with $n$ vertices, either correctly concludes that $G_1$ has a minor isomorphic to $K_h$
 or decides whether $G_1$ is isomorphic to $G_2$ in time $n^{\CO((\log h)^{c})}$ for some absolute constant $c$.
\end{theorem}

\begin{proof}
 We present a recursive algorithm that,
 given two vertex-colored graphs $(G_1,\chi_1)$ and $(G_2,\chi_2)$ and a color $c_0$ such that for $S_i \coloneqq \chi_i^{-1}(c_0)$ it holds that $|S_i| < h$ and $G_i - S_i$ is connected for $i=1,2$,
 either correctly concludes that $G_1$ has a minor isomorphic to $K_h$ or computes a representation for $\Iso((G_1,\chi_1),(G_2,\chi_2))[S_1]$.
 The color $c_0$ does not have to be in the range of the $\chi_i$ (we set $\chi_i^{-1}(c_0) = \emptyset$ in this case).
 Thus initializing it with a color $c_0$ not in the range, we have $|S_i| = 0 < h$,
 in which case the algorithm simply decides whether $\Iso((G_1,\chi_1),(G_2,\chi_2)) \neq \emptyset$,
 that is, decides whether $(G_1,\chi_1)$ and $(G_2,\chi_2)$ are isomorphic.
 (For $S_1 = S_2 = \emptyset$, we define $\Iso((G_1,\chi_1),(G_2,\chi_2))[S_1]$ to contain the empty mapping if $(G_1,\chi_1)$ and $(G_2,\chi_2)$ are isomorphic, in the other case $\Iso((G_1,\chi_1),(G_2,\chi_2))[S_1]$ is empty.)
 
 \begin{figure}
  \centering
  \begin{tikzpicture}[scale=0.6,use Hobby shortcut]
   \draw[thick,fill=yellow, fill opacity=0.5] (0,1.8) ellipse (5cm and 4cm);
   
   \begin{scope}
   \clip (0,1.8) ellipse (5cm and 4cm);
   \draw[thick, fill=blue, fill opacity=0.3] (0,5.6) ellipse (2.5cm and 2cm);
   \end{scope}
   \node at (0,4.6) {$S_i$};
   
   \draw[thick, fill=red, fill opacity=0.3] (1.7,3) ellipse (1.8cm and 1.2cm);
   \node at (1.7,3) {$X_i$};
   
   \node at (-1.5,2) {$D_i$};

   \begin{scope}[xshift=-3.5cm,yshift=0.5cm]
   \begin{scope}[rotate=-90]
   \gurke{}
   \node at (0,-4) {$Z_{i,1}$};
   \end{scope}
   \begin{scope}[rotate=-50]
   \gurke{}
   \node at (0,-4) {$Z_{i,2}$};
   \randA{}
   \end{scope}
   \begin{scope}[rotate=-10]
   \gurke{}
   \node at (0,-4) {$Z_{i,3}$};
   \end{scope}
   \node at (0,0) {$S_{i,1}$};
   \end{scope}
   
   \begin{scope}[rotate=0,yshift=-1cm]
   \gurke{}
   \node at (0,-4) {$Z_{i,4}$};
   \node at (0,0) {$S_{i,4}$};
   \randB{}
   \end{scope}

   \begin{scope}[xshift=3.5cm,yshift=0.5cm]
   \begin{scope}[rotate=90]
   \gurke{}
   \node at (0,-4) {$Z_{i,7}$};
   \end{scope}
   \begin{scope}[rotate=50]
   \gurke{}
   \node at (0,-4) {$Z_{i,6}$};
   \randA{}
   \end{scope}
   \begin{scope}[rotate=10]
   \gurke{}
   \node at (0,-4) {$Z_{i,5}$};
   \end{scope}
   \node at (0,0) {$S_{i,5}$};
   \end{scope}
   
   \node at (2.5,-6) {$H_{i,P_2}$};
   \node at (-7,2.7) {$H_{i,P_1}$};
   \node at (7,2.7) {$H_{i,P_3}$};
   
  \end{tikzpicture}
  \caption{Visualization of the graph decomposition.}
  \label{fig:dec}
 \end{figure}
 
 So let $(G_1,\chi_1)$ and $(G_2,\chi_2)$ be the vertex-colored input graphs, and let $c_0$ be a color such that $|S_i| < h$.
 Let $t \coloneqq \lceil a^2h^3\log h\rceil \in \CO(h^3 \log h)$ (as before, $a$ denotes the constant from Theorem \ref{thm:average-degree-excluded-minor}).
 The algorithm first applies Theorem \ref{thm:initial-color} to the graph $G_i$ and the parameter $t$.
 This results in a pair-colored graph $(G_i',\chi_i')$, a color $c_i \in \{\chi_i'(v,v) \mid v \in V(G_i')\}$ and a subset $X_i = \{v \in V(G_i') \mid \chi_i'(v,v) = c_i\} \subseteq V(G_i)$ for both $i \in \{1,2\}$,
 or the algorithm correctly concludes that one of the graphs has a minor isomorphic to $K_h$.
 If $c_1 \neq c_2$ or a minor is detected in only one graph, then the input graphs are non-isomorphic.
 So suppose that $c \coloneqq c_1 = c_2$.
 Then $X_1^{\varphi} = X_2$ for every $\varphi \in \Iso(G_1,G_2)$ by Theorem \ref{thm:initial-color}.
 
 Now let $D_i \coloneqq \cl_t^{G_i}(X_i)$.
 Note that $D_1^{\varphi} = D_2$ for every $\varphi \in \Iso(G_1,G_2)$.
 Also observe that $S_i \subseteq D_i$ since $|S_i| = |\chi_i^{-1}(c_0)| < h \leq t$.
 Let $Z_{i,1},\dots,Z_{i,k}$ be the vertex sets of the connected components of $G_i - D_i$ and define $\CZ_i \coloneqq \{Z_{i,1},\dots,Z_{i,k}\}$
 (if the number of connected components differs in the two graphs then they are non-isomorphic).
 See Figure \ref{fig:dec} for a visualization.
  
 If $k = 1$ and $|D_i| < h$, then the algorithm proceeds as follows.
 First, the coloring $\chi_i$, $i \in \{1,2\}$, is updated to take membership in the set $D_i$ into account, i.e.,
 $\chi_i(v)$ is replaced by $\chi_i(v) \coloneqq (\chi_i(v),1)$ if $v \in D_i$ and $\chi_i(v) \coloneqq (\chi_i(v),0)$ if $v \in V(G_i) \setminus D_i$. 
 Afterwards, the algorithm computes a set $X_i^1$ according to the above procedure with respect to the input graphs $G_i^1\coloneqq G_i - D_i$.
 Let $D_i^1\coloneqq \cl_t^{(G_i,\chi_i)}(X_i^1)$ be the closure of $X_i^1$ in the graph $G_i$ (rather than $G_i^1$).
 Then $D_i^1 \supseteq D_i$ since $|D_i| < h \leq t$.
 Moreover, $D_i^1 \supsetneq D_i$ since $X_i^1 \subseteq D_i^1$ and $\emptyset \neq X_i^1\subseteq V(G_i^1)$.
 This procedure is repeated until $|D_i^{j^*}| \geq h$, or $k \geq 2$, or $k=0$ for some $j^*\geq 1$.
 
 So without loss of generality suppose that $|D_i| \geq h$, $k \geq 2$ or $ k = 0$.
 If $k = 0$ and $|D_i| < h$, then $V(G_i) = D_i$ and therefore $|V(G_i)|<h$, and the statement of
 the theorem can directly be obtained from Babai's quasipolynomial
 time isomorphism test \cite{Babai16} since both graphs have size at
 most $h-1$.
 
 In the following, suppose that $k \geq 2$ or $|D_i| \geq h$.
 Let $S_{i,j} \coloneqq N_{G_i}(Z_{i,j})$ for all $j \in [k]$ and $i \in \{1,2\}$.
 Note that $|S_{i,j}| < h$ by Theorem \ref{thm:small-separator-for-t-cr-bounded-closure}.
 Finally, define $H_{i,j} \coloneqq G[Z_{i,j} \cup S_{i,j}]$ and $\chi_{i,j}^{H} \colon V(H_{i,j}) \rightarrow C \times \{0,1\}$ to be the vertex-coloring defined by
 \[\chi_{i,j}^{H}(v) \coloneqq \begin{cases}
                                (\chi_i(v),1) &\text{if } v \in Z_{i,j},\\
                                (\chi_i(v),0) &\text{otherwise}
                               \end{cases}\]
 for all $j \in [k]$ and both $i \in \{1,2\}$.
 For each pair $j_1,j_2 \in [k]$ and $i_1,i_2 \in \{1,2\}$ we compute the set of isomorphisms
 \[\Phi_{j_1,j_2}^{i_1,i_2} \coloneqq \Iso((H_{i_1,j_1},\chi_{i_1,j_1}),(H_{i_2,j_2},\chi_{i_2,j_2}))[S_{i_1,j_1}]\]
 recursively.
 
 \medskip
 
 For both $i \in \{1,2\}$ we define an equivalence relation $\sim_i$ on $[k]$ via $j_1 \sim_i j_2$ if and only if $S_{i,j_1} = S_{i,j_2}$ for $j_1,j_2 \in [k]$.
 Let $\CP_i\coloneqq\{P_{i,1},\dots,P_{i,p}\}$ be the corresponding partition into equivalence classes
 (as visualized in Figure \ref{fig:dec}).
 For each $P_i \in \CP_i$ let $S_{i,P_i} \coloneqq S_{i,j}$ for some $j \in P_i$.
 (By definition, $S_{i,P_i}$ does not depend on the choice of $j \in P_i$.)
 Also, define $H_{i,P_i} \coloneqq G_i[(\bigcup_{j \in P_i} Z_{i,j}) \cup S_{i,P_i}]$ and let $\chi_{i,P_i}^{H}$ be the coloring defined by
 \[\chi_{i,P_i}^{H}(v) \coloneqq \begin{cases}
                                  (\chi_i(v),0) &\text{if } v \in S_{i,P_i},\\
                                  (\chi_i(v),1) &\text{otherwise}
                                 \end{cases}\]
 for all $v \in V(H_{i,P_i})$.
 For each $i_1,i_2 \in \{1,2\}$ and $P_1 \in \CP_{i_1}$ and $P_2 \in \CP_{i_2}$ the algorithm computes
 \[\Phi_{P_1,P_2}^{i_1,i_2} \coloneqq \Iso((H_{i_1,P_1},\chi_{i_1,P_1}^{H}),(H_{i_2,P_2},\chi_{i_2,P_2}^{H}))[S_{i_1,P_1}]\]
 as follows.
 Without loss of generality assume that $P_1 \in \CP_1$ and $P_2 \in \CP_2$.
 We formulate the isomorphism problem between $(H_{1,P_1},\chi_{1,P_1}^H)$ and $(H_{2,P_2},\chi_{2,P_2}^H)$ as an instance of multiple-labeling-coset isomorphism.
 We define another equivalence relation $\simeq$ on $P_1 \uplus P_2$  via
 \[j_1\simeq j_2 \;\;\;\Leftrightarrow\;\;\; \Phi_{j_1,j_2}^{i_1,i_2} \neq \emptyset\]
 where $j_1 \in P_{i_1}$ and $j_2 \in P_{i_2}$.

 Again, we partition $P_1 \uplus P_2 = Q_1\cup\ldots\cup Q_q$ into the equivalence classes of the relation $\simeq$.
 For each equivalence class $Q_j$ we fix one representative $j^* \in Q_j$ and pick $i^{*} \in \{1,2\}$ such that $j^{*} \in P_{i^{*}}$.
 Let $\lambda_{j^*} \colon S_{i^{*},P_{i^{*}}}\to[|S_{P_{i^{*}}}|]$ be an arbitrary bijection.
 
 Let $i \in \{1,2\},j_i\in P_i\cap Q_j$ and define $\rho_{j_i}\Gamma_{j_i} \coloneqq\Phi_{j_i,j^*}^{i,i^{*}}\lambda_{j^*}$.
 Moreover, we define $\CX_{i,P_i} \coloneqq (S_{i,P_i},L_{i,P_i},\Fp_{i,P_i})$ where
 \[L_{i,P_i} \coloneqq \{\rho_{j_i}\Gamma_{j_i}\mid j_i \in P_i\}\] and
 \[\Fp_{i,P_i}(\rho_{j_i}\Gamma_{j_i})\coloneqq\{\!\!\{j\mid j_i'\in P_i\cap Q_j\text{ and }\rho_{j_i}\Gamma_{j_i} = \rho_{j_i'}\Gamma_{j_i'}\}\!\!\}\]
 (for each $j_i'$ such that $\rho_{j_i}\Gamma_{j_i} = \rho_{j_i'}\Gamma_{j_i'}$
 the element $j$ is added to the multiset where $j_i' \in P_i \cap Q_j$).
 \begin{claim}
  $\Phi_{P_1,P_2}^{1,2} = \Iso(\CX_{1,P_1},\CX_{2,P_2})$.
 \end{claim}
 \begin{claimproof}
  Let $\varphi \in \Iso((H_{1,P_1},\chi_{1,P_1}^{H}),(H_{2,P_2},\chi_{2,P_2}^{H}))$ and let $\sigma \colon P_1 \rightarrow P_2$ be the unique bijection such that $Z_{1,j}^{\varphi} = Z_{2,\sigma(j)}$ for all $j \in P_1$.
  Let $j_1 \in P_1$ and consider the labeling coset $\rho_{j_1}\Gamma_{j_1} \in L_{1,P_1}$.
  Let $j_2 \coloneqq \sigma(j_1)$.
  Then $j_1 \simeq j_2$ since $\varphi[Z_{1,j_1}] \in \Iso((H_{1,j_1},\chi_{1,j_1}),(H_{2,j_2},\chi_{2,j_2}))$.
  Let $j^{*}=j_1^*=j_2^*$ be the representative from the equivalence class containing $j_1$ and $j_2$ and pick $i^{*} \in \{1,2\}$ such that $j^{*} \in P_{i^{*}}$.
  Then
  \[\varphi[S_{1,j_1}] \Phi_{j_2,j^{*}}^{2,i^{*}} = \Phi_{j_1,j^{*}}^{1,i^{*}}.\]
  Since $\lambda_{j_1*}=\lambda_{j_2^*}$, this implies that $(\varphi[S_{1,j_1}])^{-1}\rho_{j_1}\Gamma_{j_1} = \rho_{j_2}\Gamma_{j_2}$ and, since the above statement holds for all $j_1 \in P_1$,
  it also means that $\Fp_{1,P_1}(\rho_{j_1}\Gamma_{j_1}) = \Fp_{2,P_2}(\rho_{j_2}\Gamma_{j_2})$ (i.e., equality between labeling cosets is preserved by the mapping $\sigma$).
  
  For the backward direction let $\varphi \in \Iso(\CX_{1,P_1},\CX_{2,P_2})$.
  This means, there is a bijection $\sigma\colon P_1 \rightarrow P_2$ such that
  \begin{enumerate}[label = (\alph*)]
   \item $j_1 \simeq \sigma(j_1)$, and
   \item $\varphi^{-1}\rho_{j_1}\Gamma_{j_1} = \rho_{\sigma(j_1)}\Gamma_{\sigma(j_1)}$
  \end{enumerate}
  for all $j_1 \in P_1$.
  This means that, for every $j_1 \in P_1$, it holds that $\varphi \in \Phi_{j_1,j_2}^{1,2}$.
  But this implies that $\varphi \in \Phi_{P_1,P_2}^{1,2}$.
 \end{claimproof}
 
 Hence, $\Phi_{P_1,P_2}^{1,2}$ can be computed using Theorem \ref{thm:multiple-labeling-cosets-isomorphism}. 
 Next, the algorithm turns to computing $\Iso((G_1,\chi_1,v_1),(G_2,\chi_2,v_2))[D_1]$ from the sets $\Phi_{P_1,P_2}^{i_1,i_2}$, $i_1,i_2 \in \{1,2\}$ and $P_1 \in \CP_{i_1}$, $P_2 \in \CP_{i_2}$.
 
 \medskip
 
 Let $v_1 \in X_1$ be an arbitrary vertex.
 For all $v_2 \in X_2$ the algorithm computes a representation of all
 isomorphisms $\varphi \in \Iso((G_1,\chi_1),(G_2,\chi_2))[D_1]$ such
 that $\varphi(v_1) = \varphi(v_2)$ as described below.
 The output of the algorithm is the union of all these isomorphisms iterating over all $v_2 \in X_2$.
 Additionally, all mappings are restricted to $S_1$ (recall that $S_1 \subseteq D_1$).
 
 Let $D_i' \coloneqq \cl_t^{(G_i',\chi_i')}(v_i)$ for $v_i\in X_i$ (and recall that $D_i = \cl_t^{G}(X_i)$).
 The algorithm first computes a $\mgamma_t$-group $\Gamma \leq \Sym(D_1')$ and a bijection $\gamma\colon D_1' \rightarrow D_2'$ such that
 \[\Iso((G_1',\chi_1',v_1),(G_2',\chi_2',v_2))[D_1'] \subseteq \Gamma\gamma\]
 using Theorem \ref{thm:compute-isomorphisms-t-closure}.
 Note that $X_i \subseteq D_i'$ by Theorem \ref{thm:initial-color}.
 For ease of notation define $\Lambda \coloneqq \Gamma\gamma[X_1]$ (observe that $X_1^{\gamma'} = X_2$ for all $\gamma' \in \Gamma\gamma$).
 In a second step, the algorithm computes another $\mgamma_t$-group $\Delta \leq \Sym(D_1)$ and a bijection $\delta\colon D_1 \rightarrow D_2$ such that
 \begin{align*}
                &\Iso((G_1,\chi_1,v_1),(G_2,\chi_2,v_2))[D_1]\\
  =         \;\;&\{\varphi \in \Iso((G_1,\chi_1),(G_2,\chi_2)) \mid \varphi[X_1] \in \Lambda\}[D_1]\\
  \subseteq \;\;&\Delta\delta
 \end{align*}
 again using Theorem \ref{thm:compute-isomorphisms-t-closure}.
 
 To compute the set of isomorphisms, we now formulate the isomorphism problem between $(G_1,\chi_1,v_1)$ and $(G_2,\chi_2,v_2)$ as an instance of coset-labeled hypergraph isomorphism.
 Let $\CH_i \coloneqq (D_i,\CE_i,\Fp_i)$ where
 \[\CE_i \coloneqq E(G_i[D_i]) \cup \{S_{i,P_i} \mid P_i \in \CP_i\} \cup \{\{v\} \mid v \in D_i\}.\]
 The function $\Fp_i$ is defined separately for all three parts of the set $\CE_i$ (if an element occurs in more than one set of the union, the colors defined with respect to the single sets are combined by concatenating them in a tuple).

 For an edge $vw \in E(G_i[D_i])$ we define $\Fp_i(vw) \coloneqq (\rho_{v,w}\Sym([2]),0)$ with the bijection $\rho_{v,w}\colon \{v,w\} \rightarrow \{1,2\}$ where $\rho_{v,w}(v) = 1$ and $\rho_{v,w}(w) = 2$.
 
 In order to define $\Fp_i$ for sets $S_{i,P_i}$, $P_i \in \CP_i$, we first define an equivalence relation $\approx$ on the disjoint union $\CP_1 \uplus \CP_2$ where $P \approx Q$ if $\Iso((H_{i_1,P},\chi_{i_1,P}^{H}),(H_{i_2,Q},\chi_{i_2,Q}^{H})) \neq \emptyset$ for $P \in \CP_{i_1}$ and $Q \in \CP_{i_2}$.
 Let $\CQ_1,\dots,\CQ_r$ be the equivalence classes.
 For each equivalence class $\CQ_j$ we fix one representative $Q_j^* \in \CQ_j$ and
 pick $i^* \in \{1,2\}$ such that $Q_j^* \in \CP_{i^{*}}$.
 Let $\rho_{Q_j^*} \colon S_{Q_j^*} \rightarrow [| S_{Q_j^*}|]$ be an arbitrary bijection.
 Let $i \in \{1,2\}$, $P_i \in \CP_i \cap \CQ_j$ and define
 \[\rho_{i,P_i}\Gamma_{i,P_i}\coloneqq \Iso((H_{i,P_i},\chi_{i,P_i}^{H}),(H_{i^{*},Q_j^*},\chi_{i^{*},Q_j^*}^{H}))[S_{i,P_i}]\rho_{Q_j^*}.\]
 Now, for $P_i \in \CP_i \cap \CQ_j$, we define
 \[\Fp_i(S_{i,P_i}) \coloneqq (\rho_{i,P_i}\Gamma_{i,P_i},j).\]
 (Intuitively speaking, each separator $S_{i,P_i}$ is associated with a color $j$ and a labeling coset $\rho_{i,P_i}\Gamma_{i,P_i}$.
 The color $j$ encodes the isomorphism type of the graph $H_{i,P_i}$ whereas the labeling coset determines which mappings between separators extend to isomorphisms between the corresponding graphs below the separators.)
 
 Finally, for $v \in D_i$, we define $\Fp_i(v) \coloneqq (v \mapsto 1,\chi_i(v) + r)$ (recall that $r$ denotes the number of equivalence classes $\CQ_1,\dots,\CQ_r$).
 Then
 \[\Iso((G_1,\chi_1,v_1),(G_2,\chi_2,v_2))[D_1] = \Iso_{\Delta\delta}(\CH_1,\CH_2)\]
 which can be computed in the desired time by Theorem \ref{thm:coset-labeled-hypergraphs-gamma-d}.
 
 This completes the description of the algorithm.
 The correctness follows from the statements made throughout the description of the algorithm.
 So it remains to analyze the running time.
 
 First observe that the number of recursive calls the algorithm performs is at most quadratic in the number of vertices of the input graphs.
 One way to see this is to associate vertices from $D_i \setminus S_i$ with the graph $G_i$, and observe that every vertex can be associated with at most one subgraph of the original input graph $G_i$ considered in some recursive call.
 So the number of subgraphs of $G_i$ considered in recursive calls is at most the number of vertices, which means that the total number of recursive calls is at most quadratic in the number of vertices of the input graphs.
 
 Also, $|\CP_i|\leq n$ and $|S_{i,j}|<h$ for both $i \in \{1,2\}$ and all $j\in[k]$.
 Hence, the computation of all sets $\Phi_{P_1,P_2}^{i_1,i_2}$, $P_1 \in \CP_{i_1}$ and $P_2 \in \CP_{i_2}$, requires time $n^{\CO((\log h)^{c})}$ by Theorem \ref{thm:multiple-labeling-cosets-isomorphism}.
 Next, the algorithm iterates over all vertices $v_2 \in X_2$ and computes isomorphisms between coset-labeled hypergraphs using Theorem \ref{thm:coset-labeled-hypergraphs-gamma-d}.
 In total, the algorithm from Theorem \ref{thm:coset-labeled-hypergraphs-gamma-d} is applied $|X_2| \leq n$ times and a single execution requires time $n^{\CO((\log h)^{c})}$.
 Overall, this gives the desired bound on the running time.
\end{proof}

We remark that, by standard reduction techniques, there is also an algorithm computing a representation for the set $\Iso(G_1,G_2)$ in time $n^{\CO((\log h)^{c})}$ assuming $G_1$ excludes $K_h$ as a minor.

We also remark that the proof of the last theorem reveals some insight into the structure of the automorphism group of a graph that excludes $K_h$ as a minor.

Let $G$ be a graph.
A \emph{tree decomposition} for $G$ is a pair $(T,\beta)$ where $T$ is a rooted tree and $\beta\colon V(T) \rightarrow 2^{V(G)}$ such that
\begin{enumerate}
 \item[(T.1)] for every $e \in E(G)$ there is some $t \in V(T)$ such that $e \subseteq \beta(t)$, and
 \item[(T.2)] for every $v \in V(G)$ the graph $T[\{t \in V(T) \mid v \in \beta(t)\}]$ is non-empty and connected.
\end{enumerate}
The \emph{adhesion-width} of $(T,\beta)$ is $\max_{t_1t_2 \in E(T)} |\beta(t_1) \cap \beta(t_2)|$.

Let $v \in V(G)$.
Also, recall that $(\Aut(G))_v = \{\varphi \in \Aut(G) \mid v^{\varphi} = v\}$ denotes the subgroup of the automorphism group of $G$ that stabilizes the vertex $v$.

\begin{theorem}
 \label{thm:structure-aut}
 Let $G$ be a graph that excludes $K_h$ as a minor.
 Then there is an isomorphism-invariant tree decomposition $(T,\beta)$ of $G$ such that
 \begin{enumerate}
  \item the adhesion-width of $(T,\beta)$ is at most $h-1$, and
  \item for every $t \in V(T)$ there is some $v \in \beta(t)$ such that $(\Aut(G))_v[\beta(t)] \in \mgamma_d$ for $d \coloneqq \lceil a^2h^3\log h\rceil$.
 \end{enumerate}
\end{theorem}

The theorem readily follows from the same arguments used to prove Theorem \ref{theo:iso}.
Indeed, consider the recursion tree $T$ of the algorithm from Theorem \ref{theo:iso} on input $(G,G)$ where each node $t \in V(T)$ is associated with the corresponding set $\beta(t) \coloneqq D_1$ (see also Figure \ref{fig:dec}).
For $t \in V(T)$ let $v \in X_1 \setminus S_1 \subseteq D_1$ (recall that $S_1 = \beta(s) \cap \beta(t)$ where $s$ is the unique parent node of $t$).
Then $D_1^{\gamma} = D_1$ for all $\gamma \in (\Aut(G))_v$ and $(\Aut(G))_v[\beta(t)] \in \mgamma_d$.
Finally, observe that $X_1 \setminus S_1 \neq \emptyset$ (in a situation where $X_1 \subseteq S_1$, it also holds that $D_1 \subseteq S_1$ and the algorithm from Theorem \ref{theo:iso} would recompute a set $X_i^{1}$).

\section{Conclusion}

We presented an isomorphism test for graph classes that exclude $K_h$ as a minor running in time $n^{\polylog(h)}$.
The algorithm builds on group-theoretic methods from \cite{Neuen22b,Wiebking20} as well as novel insights on the isomorphism-invariant structure of graphs excluding the minor $K_h$.

In a follow-up work \cite{Neuen22}, the second author could show that
most of the results presented in this work can be extended to graph
classes that only exclude $K_h$ as a topological subgraph.  Actually,
most of the techniques developed here already extend to classes that
only exclude $K_h$ as a topological subgraph rather than as a minor.
In particular, this includes Theorem \ref{thm:small-separator-for-t-cr-bounded-closure}.
The only part of our algorithm that exploits closure under taking minors is the
subroutine from Theorem \ref{thm:initial-color} which provides the
initial set $X$ together with sufficient structural information on
this set. It turns out that this theorem can also be extended to graphs
only excluding $K_h$ as a topological subgraph, but this comes at the
price of a much more complicated analysis which also builds on
different tools.

In another related work \cite{GroheSW20}, Schweitzer together with the first and third author of this paper investigates the structure of automorphism groups of graphs excluding $K_h$ as a minor in more detail.
This work confirms a conjecture of Babai \cite{Babai81} stating that all composition factors of such groups are cyclic groups, alternating groups, or their size is bounded by $f(h)$ for some function $f$.
Observe that our structural insights summarized in Theorem \ref{thm:structure-aut} do not imply such a statement, since the restrictions only take effect after individualizing some vertex.
On the other hand, Theorem \ref{thm:structure-aut} provides polynomial bounds on the complexity of the composition factors (after individualizing some vertex) whereas the bounds from \cite{GroheSW20} may depend arbitrarily on $h$. 

Yet another recent result related to this work by Lokshtanov et al.\ \cite{LokshtanovPPS22} shows that the graph isomorphism problem is also fixed-parameter tractable (i.e., it can be solved in time $f(h) \cdot n^c$) when parameterized by the Hadwiger number (the maximum $h$ such that $K_h$ is a minor).
Note that our result is independent of this fpt result, because our algorithm is obviously not fpt, but it also has no exponential dependence on $h$ (in fact, the function $f$ obtained in \cite{LokshtanovPPS22} may not even be computable).
Running times of the form $n^{\polylog(k)}$ for parameterized problems with input size $n$ and parameter $k$ so far seem to be quite specific to the isomorphism problem.
It may be worthwhile to study them more systematically in a broader context.
More specifically, looking at the isomorphism problem, a natural question concerns the existence of isomorphism algorithms running in time $n^{\polylog(k)}$ for other graph parameters.
As a concrete example, can isomorphism of graphs of rank-width $k$ be tested in time $n^{\polylog(k)}$?

\bibliographystyle{plainurl}
\small
\bibliography{literature}

\end{document}